\documentclass[preprint]{sigplanconf}

\pdfoutput=1

\usepackage[utf8]{inputenc}
\usepackage{amssymb,amsmath}
\usepackage{amsthm}
\usepackage{stmaryrd,colonequals,mathabx}
\usepackage{mathtools,empheq}
\usepackage[f]{esvect}
\usepackage{fancybox}
\usepackage{ifthen}
\usepackage{hhline}
\usepackage{xargs}
\usepackage{xspace}
\usepackage{calc}
\usepackage{tikz}
\usepackage{tikz-cd}
\usepackage{pgffor}
\usepackage{array}
\usepackage{etoolbox}
\usepackage[inline]{enumitem}
\usepackage{bussproofs}
\usepackage{listings}
\usepackage{wrapfig}
\usepackage{url}

\usepackage{local-macros}
\usepackage{thm}


\setenumerate{listparindent=\parindent}

\begin{document}

\setlength{\pdfpageheight}{\paperheight}
\setlength{\pdfpagewidth}{\paperwidth}

\conferenceinfo{CONF 'yy}{Month d--d, 20yy, City, ST, Country}
\copyrightyear{20yy}
\copyrightdata{978-1-nnnn-nnnn-n/yy/mm}
\copyrightdoi{nnnnnnn.nnnnnnn}



\title{Type Theory based on Dependent Inductive and Coinductive Types}

\authorinfo{Henning Basold}
           {Radboud University \\ CWI, Amsterdam}
           {h.basold@cs.ru.nl}
\authorinfo{Herman Geuvers}
           {Radboud University \\ Technical University Eindhoven}
           {herman@cs.ru.nl}

\maketitle

\begin{abstract}
  We develop a dependent type theory that is based purely on inductive
  and coinductive types, and the corresponding recursion and
  corecursion principles.
  This results in a type theory with a small set of rules, while still being
  fairly expressive.
  For example, all well-known basic types and type formers that are needed for
  using this type theory as a logic are definable: propositional
  connectives, like falsity, conjunction, disjunction, and function space,
  dependent function space, existential quantification, equality,
  natural numbers, vectors etc.
  The reduction relation on terms consists solely of a rule for recursion and a
  rule for corecursion.
  The reduction relations for well-known types arise from that.
  To further support the introduction of this new type theory, we also
  prove fundamental properties of its term calculus.
  Most importantly, we prove subject reduction and strong normalisation of
  the reduction relation, which gives computational meaning to the terms.

  The presented type theory is based on
  ideas from categorical logic that have been investigated before by
  the first author, and it extends Hagino's categorical data types to
  a dependently typed setting.
  By basing the type theory on concepts from category theory we maintain the
  duality between inductive and coinductive types, and it allows us to describe,
  for example, the function space as a coinductive type.
\end{abstract}




\section{Introduction}
\label{sec:intro}

In this paper, we develop a type theory that is based solely on dependent
inductive and coinductive types.
By this we mean that the only way to form new types is by specifying the type
of their corresponding constructors or destructors, respectively.
From such a specification, we get the corresponding recursion and corecursion
principles.
One might be tempted to think that such a theory is relatively weak as, for
example, there is no function space type.
However, as it turns out, the function space is definable as a
coinductive type.
Other type formers, like the existential quantifier, that are needed in logic,
are definable as well.
Thus, the type theory we present in this paper encompasses intuitionistic
predicate logic.

Why do we need another type theory, especially since
Martin-Löf type theory (MLTT)~\cite{MartinLof-ModelsITT} or the calculus of
inductive constructions (CoIC)~\cite{Paulin-Mohring93,Werner1994,BertotC04}
are well-studied
frameworks for intuitionistic logic?
The main reason is that the existing type theories have no explicit
dependent coinductive types.
Gim\'enez~\cite{Gimenez-RecursiveSchemes} discusses an extension of the CoIC
with coinductive types and guarded recursive schemes but proves no properties
about the conversion relation.
On the other hand, Sacchini~\cite{Sacchini-TypeBasedProductivity} extended the
CoC with streams, and proves subject reduction and strong normalisation.
However, the problem of limited support for general coinductive types remains.
Finally, we should also mention that general coinductive types are
available in implementations like Coq~\cite{Coq:manual}, which is based
on~\cite{Gimenez-RecursiveSchemes}, Agda~\cite{Agda:system} and
Nuprl~\cite{Constable:NuprlTT}.
Yet, none of these has a formal justification, and Coq's coinductive types
are even known to have problems (e.g. related to subject reduction).

One might argue that dependent coinductive types can be encoded through
inductive types, see~\cite{Ahrens:MTypes-HoTT,Basold-DepCoindFibDialg}.
However, it is not clear whether such an encoding gives rise to a good
computation principle in an intensional type theory such as MLTT or CoIC,
see~\cite{cLab:FinalChain-TT}.
This becomes an issue once we try to prove propositions about terms of
coinductive type.

Other reasons for considering a new type theory are of foundational interest.
First, taking inductive and coinductive types as core of
the type theory reduces the number of deduction rules considerably
compared to, for example, MLTT with W- and M-types.
Second, it is an interesting fact that the (dependent) function space can be
described as a coinductive type.
This is well-known in category theory but we do not know of any treatment of
this fact in type theories.
Thus the presented type theory allows us to deepen our understanding of
coinductive types.

\paragraph{Contributions}

Having discussed the raison d'être of this paper, let us briefly mention the
technical contributions.
First of all, we introduce the type theory and show how important logical
operators can be represented in it. We also discuss some other basic examples,
including one that shows the difference to existing theories with coinductive
types.
Second, we show that computations of terms, given in form of a reduction
relation, are meaningful, in the sense that the reduction relation preserves
types (subject reduction) and that all computations are terminating (strong
normalisation).
Thus, under the propositions-as-types interpretation, our type theory can
serve as formal framework for intuitionistic reasoning.


\paragraph{Related Work}
A major source of inspiration for the setup of our type theory is categorical
logic.
Especially, the use of fibrations, brought forward in~\cite{Jacobs1999-CLTT},
helped a great deal in understanding how coinductive types should be treated.
Another source of inspiration is the view of type theories as internal language
or even free model for categories, see for example~\cite{LambekScott-HOCatLog}.
This view is especially important in topos theory, where final coalgebras
have been used as foundation for predicative, constructive
set theory~\cite{Aczel:nonwfs,vdBerg-Non-wellfoundedTrees,vanDenBerg:thesis}.
These ideas were extended in~\cite{Basold-DepCoindFibDialg},
which discusses the categorical analogue of the type theory of this paper,
see also~\iSecRef{dtcc}.

Let us briefly discuss other type theories that the present work relates to.
Especially close is the copattern calculus introduced
in~\cite{Abel2013-Copatterns}, as there the coinductive types are also
specified by the types of their destructors.
However, said calculus does not have dependent types, and it is based on
systems of equations to define terms, whereas the calculus in the present paper
is based on recursion and corecursion schemes.

To ensure strong normalisation, the copatterns have been combined with size
annotations in~\cite{Abel2013}.
Due to the nature of the reduction relation in these copattern-based calculi,
strong normalisation also ensures productivity for coinductive types or,
more generally, well-definedness~\cite{BasoldHansen:Welldef-Equiv-CoInd}.
As another way to ensure productivity, guarded recursive types were proposed
and in \cite{Bizjak:GuardedDepRecTypes} guarded recursion was extended to
dependent types.
Guarded recursive types are not only applicable to strictly positive types,
which we restrict to in this paper, but also to positive and even negative
types.
However, it is not clear how one can include inductive types into such
a type theory, which are, in the authors opinion, crucial to mathematics
and computer science.
Finally, in~\cite{Sacchini-TypeBasedProductivity} another type theory
with type-based termination conditions and a type former for streams
has been introduced.
This type theory, however, lacks again dependent coinductive types.

\paragraph{Outline}

The rest of the paper is structured as follows.
In~\iSecRef{dtcc}, we briefly discuss the ideas from category theory
that motivate the definition of the type theory.
This section is strictly optional and can be safely skipped.
The type theory itself is introduced in~\iSecRef{syntax}, and
in~\iSecRef{examples} we give a host of examples and discuss the representation
of logical operators.
After that, we justify in~\iSecRef{pre-syntax} the definition of the typing
rules of~\iSecRef{syntax} by giving an untyped version of the calculus.
Moreover, this section serves as the technical basis for the strong
normalisation proof.
\secRef{meta-prop} is devoted to proving important properties of the type
theory, namely subject reduction in~\iSecRef{subject-reduction}, and
strong normalisation in~\iSecRef{sn}.
Finally, we make concluding remarks and discuss future work
in~\iSecRef{conclusion}.


\section{Categorical Dependent Data Types}
\label{sec:dtcc}

Before we introduce the actual calculus, let us briefly describe the structure
the calculus shall capture.
This is a short recap from~\cite{Basold-DepCoindFibDialg}, to which we refer
for more details.
Note, that this section is completely optional and only serves as
motivation for those familiar with category theory.

We begin with the definition of dialgebras and associated notions,
see~\cite{Hagino-Dialg}.
\begin{definition}
  \label{def:dialgebra}
  Let $\Cat{C}$ and $\Cat{D}$ be categories and $F, G : \Cat{C} \to \Cat{D}$ be
  functors.
  An $(F, G)$-\emph{dialgebra} is morphism $d : F X \to G X$ in $\Cat{D}$ for
  an object $X$ in $\Cat{C}$.
  We say that a morphism $f : X \to Y$ is a dialgebra \emph{homomorphism}
  from the dialgebra $d : F X \to G X$ to
  $e : F Y \to G Y$, if $e \circ Ff = Gf \circ d$.
  This allows us to form the category $\DialgC{F, G}$ of dialgebras and their
  homomorphisms.
  Finally, a dialgebra is an \emph{initial} (resp. \emph{final})
  $(F,G)$-dialgebra if it is an initial (resp. final) object in $\DialgC{F, G}$,
  see~\cite{Basold-DepCoindFibDialg}.
\end{definition}

Let us discuss an example of a dialgebra in the category of sets.
\begin{example}
  Let $F, G : \SetC \to \SetC \times \SetC$ be given by
  $F = \prodArr{\T, \Id}$ and $G = \prodArr{\Id, \Id}$, that is, $F$ maps a
  set $X$ to the pair $(\T, X)$ in the product category.
  Similarly, $G$, the diagonal functor, maps $X$ to $(X, X)$.
  Now, let $z : \T \to \N$ and $s : \N \to \N$ be the constant zero map
  and the successor on natural numbers, respectively.
  It is then easy to see that $(z, s) : F (\N) \to G (\N)$ is an initial
  dialgebra.
  \qedDef
\end{example}

Initial and final dialgebras will allow us to describe dependent data types
conveniently, where the dependencies are handled through the use of fibrations.
\begin{definition}
  \label{def:fibration}
  Let $P : \TCat \to \BCat$ be a functor, where $\TCat$ is called the
  \emph{total} category and $\BCat$ the \emph{base} category.
  A morphism $f : A \to B$ in $\TCat$ is said to be \emph{cartesian over}
  $u : I \to J$ in $\BCat$, provided that
  \begin{enumerate*}[label=\roman*)]
  \item $P f = u$, and
  \item for all $g : C \to B$ in $\TCat$ and $v : PC \to I$ with
    $Pg = u \circ v$ there is a unique $h : C \to A$ such that $f \circ h = g$.
  \end{enumerate*}
  For $P$ to be a \emph{fibration}, we require that
  for every $B \in \TCat$ and $u : I \to PB$ in $\Cat{B}$, there is
  a cartesian morphism $f : A \to B$ over $u$.
  Finally, a fibration is \emph{cloven}, if it comes with a unique choice
  for $A$ and $f$, in which case we denote $A$ by $\reidx{u}B$ and
  $f$ by $\cartL{u} B$, as displayed in the diagram on the right.
\end{definition}

\begin{wrapfigure}[7]{r}{.2\textwidth}
\vspace{-1.2\baselineskip}
\hspace*{-20pt}
  \begin{tikzcd}[row sep=0.1cm, column sep=0.25cm]
    C \arrow[bend left=15]{drrr}{g}
      \arrow[dashed,shorten >= -5pt]{dr}[swap]{!h}
      & & \\
    & \reidx{u} B
      \arrow{rr}[swap]{\cartL{u} B}
    & & B
      & \TCat \arrow{dddd}{P}
    \\ \\ \\
    PC
      \arrow[bend left=15]{drrr}{P g}
      \arrow{dr}[swap]{v} & & \\
    & I
      \arrow{rr}[swap]{u} & & PB
      & \BCat
  \end{tikzcd}
\end{wrapfigure}

On cloven fibrations, we can define for each $u : I \to J$ in $\BCat$
a functor, the \emph{reindexing} along $u$, as follows.
Let us denote by $\Cat{P}_I$ the category
having objects $A$ with $P(X) = I$ and morphisms $f : A \to B$ with
$P(f) = \id_I$.
We call $\Cat{P}_I$ the \emph{fibre above $I$}.
The assignment of $\reidx{u} B$ to $B$ for a cloven fibration can then be
extended to a functor $\reidx{u} : \Cat{P}_J \to \Cat{P}_I$.
Moreover, one can show that $\reidx{\id_I} \cong \Id_{\Cat{P}_I}$ and
$\reidx{(v \circ u)} \cong \reidx{u} \circ \reidx{v}$.
In this work, we are mostly interested in \emph{split fibrations},
which are cloven fibrations such that the above isomorphisms are
equalities, that is, $\reidx{\id_I} = \Id_{\Cat{P}_I}$ and
$\reidx{(v \circ u)} = \reidx{u} \circ \reidx{v}$.

\begin{example}[See~\cite{Jacobs1999-CLTT}]
  Important examples of fibrations arise from categories with pullbacks.
  Let $\Cat{C}$ be a category and $\ArrC{\Cat{C}}$ be the arrow category with
  morphisms $f : X \to Y$ of $\Cat{C}$ as objects and commutative squares as
  morphisms.
  We can then define a functor $\cod : \ArrC{\Cat{C}} \to \Cat{C}$ by
  $\cod(f : X \to Y) = Y$.
  This functor turns out to be a fibration, the \emph{codomain fibration},
  if $\Cat{C}$ has pullbacks.
  If we are given a choice of pullbacks, then $\cod$ is cloven.

  The split variant of this construction is given by the category of
  \emph{set-indexed families} over $\Cat{C}$.
  Let $\Fam{\Cat{C}}$ be the category that has families $\setFam{X_i}{i \in I}$
  of objects $X_i$ in $\Cat{C}$ indexed by a set $I$.
  The morphisms $\setFam{X_i}{i \in I} \to \setFam{Y_j}{j \in J}$ in
  $\Fam{\Cat{C}}$ are pairs $(u, f)$ where $u : I \to J$ is a function and $f$
  is an $I$-indexed family of morphisms in $\Cat{C}$ with
  $\setFam{f_i : X_i \to Y_{u(i)}}{i \in I}$.
  It is then straightforward to show that the functor
  $p : \Fam{\Cat{C}} \to \SetC$, given by projecting on the index set,
  is a split fibration.
  \qedDef
\end{example}

To model \emph{dependent} data types, we consider dialgebras in the fibres of a
fibration $P : \TCat \to \BCat$.
Before giving a general account, let us look at an important example:
the dependent function space.
\begin{example}
  Suppose that $\Cat{B}$ has a final object $\T$, and let $I \in \Cat{B}$.
  Thus, there is a morphism $!_{I} : I \to \T$, which gives rise
  to the \emph{weakening} functor
  $\reidx{!_{I}} : \Cat{P}_{\T} \to \Cat{P}_{I}$.
  We can then show that for each $X \in \Cat{P}_I$ the dependent function space
  $\Pi_I X$ is the final $(\reidx{!_{I}}, K_X)$-dialgebra, whereby $K_X$ is
  the functor mapping every object to $X$ and morphism to $\id_X$.
  That is to say, there is a dialgebra $\ev_X : \reidx{!_{I}}(\Pi_I X) \to X$
  that evaluates a function on an argument from $I$, such that for each
  dialgebra $f : \reidx{!_{I}}(U) \to X$ there is a unique
  $\lambda f : U \to \Pi_I X$ with $\ev_X \circ \reidx{!_{I}}(\lambda f) = f$.
  \qedDef
\end{example}

From a categorical perspective, the dependent function space is actually a
functor $\Cat{P}_I \to \Cat{P}_{\T}$ that is, moreover, right adjoint to the
weakening functor : $\reidx{!_I} \dashv \Pi_I$.
To capture this, we allow data types to have parameters.
\begin{definition}
  \label{def:parameterised-dialg}
  Let $\Cat{C}$, $\Cat{D}$, $\Cat{X}$ be categories, and
  $F : \Cat{X} \times \Cat{C} \to \Cat{D}$ be a functor.
  We define a functor
  $\funCur{F} : \FunCat{\Cat{X}}{\Cat{C}} \to \FunCat{\Cat{X}}{\Cat{D}}$
  between functor categories by
  \begin{equation}
    \funCur{F}(H) = F \circ \prodArr{\Id_{\Cat{X}}, H}.
  \end{equation}
  Let $G : \Cat{X} \times \Cat{C} \to \Cat{D}$ be another functor.
  A \emph{parameterised} $(F, G)$-dialgebra is an
  $(\funCur{F}, \funCur{G})$-dialgebra, that is, a natural transformation
  $\nat{\delta}{\funCur{F}(H)}{\funCur{G}(H)}$ for a functor
  $H : \Cat{X} \to \Cat{C}$.
\end{definition}

\begin{example}
  The dependent function space $\Pi_I$ functor is a final, parameterised
  $(G, \pi_1)$-dialgebra, where
  $G, \pi_1 : \Cat{P}_I \times \Cat{P}_{\T} \to \Cat{P}_I$,
  $G = \reidx{!_I} \circ \pi_2$ and
  $\pi_1$ is the product projection.
  This is a consequence of the fact that $\pi_1(X, U) = K_X(U)$,
  $G(X, U) = \reidx{!_I}(U)$, and that for each $X$ the function
  space $\Pi_IX$ is a final $(\reidx{!_I}, K_X)$-dialgebra.
  This allows us to prove that the evaluation $\ev_X$ is natural in $X$
  and that $\Pi_I$ is final in $\DialgC{\funCur{G}, \funCur{\pi_1}}$.
  \qedDef
\end{example}

Let $P : \TCat \to \BCat$ be a cloven fibration, $I \in \BCat$ and $u$
a tuple $u = (u_1, \dotsc, u_n)$ of morphisms $u_k : J_k \to I$ in $\BCat$.
Then for every $\Cat{X}$ there is a functor
$G_u : \Cat{X} \times \Cat{P}_I \to \prod_{k=1}^n \Cat{P}_{J_k}$ given by
\begin{equation}
  \label{eq:signature-reidx-fun}
  G_u = \prodArr{\reidx{u_1}, \dotsc, \reidx{u_n}} \circ \pi_2.
\end{equation}
Now we are in the position to define what it means for a category to have
strictly positive, dependent data types.
\begin{definition}
  \label{def:dt-complete}
  Given a cloven fibration $P : \TCat \to \BCat$ we define by mutual induction
  data type completeness, the class $\SPSig$ of strictly
  positive signatures and the class $\SPDT$ of strictly positive data
  types.

  We say that $P$ is \emph{data type complete}, if for all $(F, u) \in \SPSig$
  an initial $(\funCur{F}, \funCur{G_u})$- and final
  $(\funCur{G_u}, \funCur{F})$-dialgebra exists.
  We denote their carriers by  $\mu (\funCur{F}, \funCur{G_u})$ and
  $\nu (\funCur{G_u}, \funCur{F})$, respectively.
  A pair $(F, u)$ is a \emph{strictly positive signature}, if
  $(F, u) \in \SPSig$ by the first rule in~\iFigRef{dtcc-rules}.
  Finally, a \emph{strictly positive data type} is a functor $F \in \SPDT$,
  as given by the other rules in~\iFigRef{dtcc-rules}.

  \begin{figure*}
    \centering
    \begin{equation*}
      \AxiomC{$\Cat{D} = \prod_{i = 1}^n \Cat{P}_{J_i}$}
      \AxiomC{$F \in \SPDT[\Cat{C} \times \Cat{P}_I \to \Cat{D}]$}
      \AxiomC{$u = (u_1 : J_1 \to I, \dotsc, u_n : J_n \to I)$}
      \TrinaryInfC{$(F, u) \in \SPSig[\Cat{C} \times \Cat{P}_I \to \Cat{D}]$}
      \DisplayProof
    \end{equation*}

  \begin{gather*}
    \AxiomC{$A \in \Cat{P}_J$}
    \UnaryInfC{$K_A^{\Cat{P}_I} \in \SPDT[\Cat{P}_I \to \Cat{P}_J]$}
    \DisplayProof
    \quad
    \AxiomC{$\Cat{C} = \prod_{i=1}^n \Cat{P}_{I_i}$}
    \UnaryInfC{$\pi_k \in \SPDT[\Cat{C} \to \Cat{P}_{I_k}]$}
    \DisplayProof
    \quad
    \AxiomC{$f : J \to I \text{ in } \Cat{B}$}
    \UnaryInfC{$\reidx{f} \in \SPDT[\Cat{P}_I \to \Cat{P}_J]$}
    \DisplayProof
    \quad
    \AxiomC{$F_i \in \SPDT[\Cat{P}_I \to \Cat{P}_{J_i}]$}
    \AxiomC{$i = 1,2$}
    \BinaryInfC{$\prodArr{F_1, F_2}
      \in \SPDT[\Cat{P}_I \to \Cat{P}_{J_1} \times \Cat{P}_{J_2}]$}
    \DisplayProof
    \\[7pt]
    \AxiomC{$F_1 \in \SPDT[\Cat{P}_I \to \Cat{P}_K]$}
    \AxiomC{$F_2 \in \SPDT[\Cat{P}_K \to \Cat{P}_J]$}
    \BinaryInfC{$F_2 \circ F_1 \in \SPDT[\Cat{P}_I \to \Cat{P}_J]$}
    \DisplayProof
    \quad
    \AxiomC{$(F, u) \in \SPSig[\Cat{C} \times \Cat{P}_I \to \Cat{D}]$}
    \UnaryInfC{$\mu (\funCur{F}, \funCur{G_u}) \in \SPDT[\Cat{C} \to \Cat{P}_I]$}
    \DisplayProof
    \quad
    \AxiomC{$(F, u) \in \SPSig[\Cat{C} \times \Cat{P}_I \to \Cat{D}]$}
    \UnaryInfC{$\nu (\funCur{G_u}, \funCur{F}) \in \SPDT[\Cat{C} \to \Cat{P}_I]$}
    \DisplayProof
  \end{gather*}
    \caption{Closure rules for data type complete categories}
    \label{fig:dtcc-rules}
  \end{figure*}
\end{definition}


\section{Typed Syntax}
\label{sec:syntax}

We introduce our type theory through its typing rules, following the categorical
syntax just given.
All definitions of this section are given by mutual induction, which we justify
in \iSecRef{pre-syntax}.

Before we formally introduce the typing rules, let us give an informal
overview of the syntax.
First of all, we will have two kinds of variables: type constructor variables
and term variables.
This leads us to use well-formedness judgements of the form
\begin{equation*}
  \validTypeN{A}[]{i},
\end{equation*}
which states that $A$ is a type in the type constructor variable context $\tyCtx$
and the term variable context $\objCtx_1$.

The type constructor variables in $\tyCtx$ are meant to capture types with terms
as parameters (dependent types), thus we need a means to deal with these
parameter dependencies.
The way we chose to do this here is by introducing \emph{parameter contexts}
and \emph{instantiations} thereof.
So we generalise the above judgement to
\begin{equation*}
  \validTypeN{A}{i},
\end{equation*}
in which $A$ is a type constructor in the combined context
$\ctx{\tyCtx}{\objCtx_1}$ with parameters in $\objCtx_2$.
Suppose that $\objCtx_2 = x_1 : B_1, \dotsc, x_n : B_n$ and that we
are given
$\typed[\objCtx_1]{t_k}{B_k \substN{t_1/x_1, \dotsc, t_{k-1}/x_{k-1}}}$,\footnote{
  Read: In the term variable context $\objCtx_1$, $t_k$ is a term of type
  $B_k$.}
then the instantiation of $A$ with these terms is denoted by
\begin{equation*}
  \validTypeN{A \inst t_1 \inst \dotsb \inst t_n}[]{i}.
\end{equation*}
Note, however, that the arrow $\ctxTo$ is \emph{not} meant to be the function
space in a higher universe, rather parameter contexts are a syntactic tool to
deal elegantly with parameters of type constructors.
We illustrate this with a small example.
Let $\objCtx_2$ and $t_1, \dotsc, t_n$ be as above,
and let $X$ be a type constructor variable.
The type system will allow us to form the judgement
\begin{equation*}
  \validTypeN[X : {\pType[\objCtx_2]}][\objCtx_1]{X}[\objCtx_2]{i},
\end{equation*}
and then instantiate $X$ with the terms $t_1, \dotsc, t_n$ to obtain
\begin{equation*}
  \validTypeN[X : {\pType[\objCtx_2]}][\objCtx_1]
  {X \inst t_1 \inst \dotsb \inst t_n}[]{i}.
\end{equation*}
Besides parameter instantiation, we also allow variables to be moved from the
term variable context into the parameter context by parameter abstraction.
Through these two mechanisms we can deal smoothly with type constructor
variables, which are dependent types with parameters.
As an example we will be able to form
\begin{equation*}
  \validTypeN[X : {\pType[(x : B, y : B)]}][\emptyset]
  {(z).(X \inst z \inst z)}[(z:B)]{i}.
\end{equation*}

Similar to type constructors with parameters, we also have terms with
parameters, and instantiations for them.
A term with parameters will be typed by a \emph{parameterised type} of the
shape $\pTerm[\objCtx_2]{A}$.
A term $s$ with $\typed[\objCtx_1]{s}{\pTerm[\objCtx_2]{A}}$ can be instantiated
with arguments, just like type constructors:
If $\objCtx_2 = x_1 : B_1, \dotsc, x_n : B_n$ and
$\typed[\objCtx_1]{t_k}{B_k \substN{t_1/x_1, \dotsc, t_{k-1}/x_{k-1}}}$ for
$1\leq k \leq n$, then
\begin{equation*}
  \objCtx_1  \vdash s \inst t_1 \inst \dotsb \inst t_n: A[\vec{t}/\vec{x}],
\end{equation*}
where $ A[\vec{t}/\vec{x}]$ denotes the simultaneous substitution of the $t_k$
for the term variables $x_k$.
In the case of terms, however, we do not allow parameter abstraction.
We will rather be able to define the (dependent) function space as
coinductive type, thus we do not need an explicit type constructor for it
and also no explicit $\lambda$-abstraction.

Having set up how we deal with type constructor variables, we come to the heart
of the calculus.
Since it shall have inductive and coinductive types, we give ourselves
type constructors that resemble the initial and final dialgebras for strictly
positive signatures in~\iSecRef{dtcc}.
These type constructors are written as
\begin{equation*}
  \indTy{X}{\objCtx}{\vec{\sigma}}{\vec{A}}
  \qquad \text{and} \qquad
  \coindTy{X}{\objCtx}{\vec{\sigma}}{\vec{A}},
\end{equation*}
where $\vec{\sigma} = \sigma_1, \dotsc, \sigma_n$ are
tuples of terms, which we will use for substitutions, and
$\vec{A} = A_1, \dotsc, A_n$ are types with a free type constructor variable
$X$.
In view of the categorical development, the $\sigma_k$ are the analogue of the
morphisms $u_k$ in the base category that were used for reindexing,
and the types $A_k$ correspond to the projections of the functor $F$.
Thus $(\vec{A}, \vec{\sigma})$ will take the role of a strictly positive
signature in the type theory.

Accordingly, we will associate constructors and a recursion scheme to inductive
types, and destructors and a
corecursion scheme to coinductive types.
Suppose, for example, that we have
$\validTypeN[X : {\pType[\objCtx]}][\objCtx_k]{A_k}[]{i}$,
$\objCtx = x_1 : B_1, \dotsc, x_n : B_n$
and $\sigma_k = (t_1, \dotsc, t_n)$.
The $k$th constructor of $\indTy{X}{\objCtx}{\vec{\sigma}}{\vec{A}}$
will have the type, using the shorthand
$\mu = \indTy{X}{\objCtx}{\vec{\sigma}}{\vec{A}}$,
\begin{equation*}
  \typed[]{\inMu_k}{
    \pTerm[\left(\objCtx_k, z : A_k \tySubst{\mu}{X} \right)]{
      \left( \mu \inst t_1 \inst \dotsb \inst t_m \right)}}.
\end{equation*}
$\inMu_k$ can now be instantiated according to the parameter context:
Suppose, for simplicity, that $\objCtx_k = y : C$ for some type $C$
that does not depend on $X$,\footnote{This is the strict positivity condition
  we will impose.} and that we are
given a term $\typed{s}{C}$.
For a recursive argument
$\typed{u}{A_k \tySubst{\mu}{X} \subst{s}{y}}$, we obtain
$\typed{\inMu_k \inst s \inst u}{
  (\mu \inst t_1 \inst \dotsb \inst t_m )\subst{s}{y}}$.

The rest of the type and term constructors are the standard rules one would
then expect.
It should also be noted that there is a strong similarity in the use
of destructors for coinductive types to those in the copattern language
in~\cite{Abel2013-Copatterns}.
Moreover, the definition scheme for generalised abstract data types
in~\cite{HamanaFiore-GADT} describes the same \emph{inductive}
types.


We now define the well-formed types and terms of the calculus through several
judgements, each of which has its own set of derivations rules.
It is understood that the derivability of these judgments is defined by
simultaneous induction.
So, Definitions \ref{def:wffctx}, \ref{def:ctxmorph}, \ref{def:wfftypecons} and
\ref{def:wfterms} should be seen as one simultaneous definition.

It is assumed that we are given two disjoint, countably infinite sets $\var$ and
$\tyVar$ of term variables and type constructor variables.
Term variables will be denoted by $x, y, z, \dotsc$, whereas type constructor
variables are denoted by capital letters $X, Y, Z, \dotsc$.
The judgements we are going to use are are the following.
\begin{itemize}
\item $\isTyCtx{\tyCtx}$ -- The type constructor variable context $\tyCtx$ is
  well-formed.
\item $\isCtx{\objCtx}$ -- The term variable context $\objCtx$ is well-formed.
\item $\validTypeN[!][\objCtx_1]{A}[\objCtx_2]{i}$
  -- The type constructor $A$ is well-formed in the combined context $\ctx{\tyCtx}{\objCtx_1}$
  and can be \emph{instantiated} with terms according to the
  \emph{parameter context} $\objCtx_2$, where it
  is implicitly assumed that $\tyCtx$, $\objCtx_1$ and $\objCtx_2$ are
  well-formed.
\item $\typed[\objCtx_1]{t}{\pTerm[\objCtx_2]{A}}$ -- The term $t$ is
  well-formed  in the term variable context $\objCtx_1$ and, after instantiating
  it with arguments according to parameter context $\objCtx_2$, is of type $A$
  with the arguments substituted into $A$.
\item $\ctxMor{\sigma}{\objCtx_1}{\objCtx_2}$ -- The context morphism $\sigma$
  is a well-formed substitution for $\objCtx_2$ with terms in context
  $\objCtx_1$.
\end{itemize}

\begin{definition}[Well-formed contexts]
  \label{def:wffctx}
  The judgements for singling out well-formed contexts (type variable contexts
  and term variable contexts) are given by the following rules.
  \begin{gather*}
    \AxiomC{}
    \UnaryInfC{$\isTyCtx{\emptyset}$}
    \bottomAlignProof
    \DisplayProof
    \quad
    \AxiomC{$\isTyCtx{\tyCtx}$}
    \AxiomC{$\isCtx{\objCtx}$}
    \BinaryInfC{$
      \isTyCtx{\tyCtx, X : \pType[\objCtx]}$}
    \bottomAlignProof
    \DisplayProof
    \\
    \AxiomC{}
    \UnaryInfC{$\isCtx{\emptyset}$}
    \bottomAlignProof
    \DisplayProof
    \quad
    \AxiomC{$\validTypeN[\emptyset][\objCtx]{A}[]{i}$}
    \UnaryInfC{$\isCtx{\objCtx, x:A}$}
    \bottomAlignProof
    \DisplayProof
  \end{gather*}
\end{definition}

\begin{remark}
  It is important to note that whenever a term variable declaration is added
  into the context, its type is not allowed to have any free type constructor
  variables, which ensures that all types are strictly positive.
  For example, we are not allowed to form the term context
  $\objCtx = x : X$ in which $X$ occurs freely.
  This prevents us, as we will see in \iExRef{pi-encode}, from forming function
  spaces $X \to A$.
  \qedDef
\end{remark}

Besides the usual notion of context, we also use \emph{parameter contexts},
to bind arguments for which no free variable exists.
We borrow the notation from the built-in dependent function space of Agda,
only changing the regular arrow used there into $\ctxTo$ to emphasise that in
our calculus this is \emph{not} the function space.

\begin{definition}[Context Morphism]
  \label{def:ctxmorph}
  We introduce the notion of context morphisms as a shorthand notation
  for sequences of terms.
  Let $\objCtx_1$ and $\objCtx_2$ be contexts.
  A \emph{context morphism} $\ctxMor{\sigma}{\objCtx_1}{\objCtx_2}$ is given by
  the following two rules.
  \begin{gather*}
    \AxiomC{}
    \UnaryInfC{$\ctxMor{\emptyCtxMor}{\objCtx_1}{\emptyset}$}
    \bottomAlignProof
    \DisplayProof
    \quad
    \AxiomC{$\ctxMor{\sigma}{\objCtx_1}{\objCtx_2}$}
    \AxiomC{$\typed[\objCtx_1]{t}{A [\sigma]}$}
    \BinaryInfC{$\ctxMor{(\sigma, t)}{\objCtx_1}{(\objCtx_2, x : A)}$}
    \bottomAlignProof
    \DisplayProof
  \end{gather*}
  where $\validType[\ctx{\emptyset}{\objCtx_2}]{A}{\univ{i}}$,
  and $A [\sigma]$ denotes the simultaneous substitution of the terms in $\sigma$
  for the corresponding variables, which is often also denoted by
  $A [\sigma] = A \subst{\sigma}{\vec{x}}$.
  \qedDef
\end{definition}

\begin{definition}[Well-formed Type Constructor]
  \label{def:wfftypecons}
  The judgement for type constructors is given inductively by the following
  rules, where it is understood that all involved contexts are well-formed.
  \begin{empheq}[box=\fbox]{gather*}
    \AxiomC{}
    \RightLabel{\TTyI}
    \UnaryInfC{$\validTypeN[][]{\oneT}[]{i}$}
    \DisplayProof
    \\[7pt]
    \AxiomC{$\isTyCtx{\tyCtx}$}
    \AxiomC{$\isCtx{\objCtx}$}
    \RightLabel{\TyVarI}
    \BinaryInfC{$
      \validType[\tyCtx, X : {\pType[\objCtx]} \mid \emptyset]{X}{
        \pType[\objCtx]}$}
    \DisplayProof
    \\[7pt]
    \AxiomC{$\validTypeN{A}{i}$}
    \AxiomC{$\isCtx{\objCtx}$}
    \RightLabel{\TyVarWeak}
    \BinaryInfC{$\validTypeN[\tyCtx, X : {\pType[\objCtx]}]{A}{i}$}
    \DisplayProof
    \\[7pt]
    \AxiomC{$\validTypeN{A}{i}$}
    \AxiomC{$\validTypeN{B}[]{i}$}
    \RightLabel{\TyWeak}
    \BinaryInfC{$\validTypeN[\tyCtx][\objCtx_1, x : B]{A}{i}$}
    \DisplayProof
    \\[7pt]
    \AxiomC{$\validType[\ctx{\tyCtx}{\objCtx_1}]
      {A}{\pType[(x : B, \objCtx_2)]}$}
    \AxiomC{$\typed[\objCtx_1]{t}{B}$}
    \RightLabel{\TyInst}
    \BinaryInfC{$\validTypeN{A \inst t}[\objCtx_2 \subst{t}{x}]{i}$}
    \DisplayProof
    \\[7pt]
    \AxiomC{$\validTypeN[\tyCtx][\objCtx_1, x : A]{B}{i}$}
    \RightLabel{\PAbstr}
    \UnaryInfC{$
      \validType[\ctx{\tyCtx}{\objCtx_1}]
      {\pAbstr{x}{B}}{\pType[(x : A, \objCtx_2)]}$}
    \DisplayProof
    \\[7pt]
    \AxiomC{$\ctxMor{\sigma_k}{\objCtx_k}{\objCtx}$}
    \AxiomC{$
      \validType[\ctx{\tyCtx, X : {\pType[\objCtx]}}{\objCtx_k}]
      {A_k}{\univ{i}}$}
    \RightLabel{\FPTy}
    \BinaryInfC{$\validTypeN[!][\emptyset]{
        \recT{\rho}{X : \pType[\objCtx]}{\vec{\sigma}}{\vec{A}}}[\objCtx]{i}$}
    \DisplayProof
  \end{empheq}
  where in the \FPTy-rule, $\rho \in \{\mu, \nu\}$,
  $\vec{\sigma} = (\sigma_1, \dotsc, \sigma_n)$,
  $\vec{A} = (A_1, \dotsc, A_n)$,
  and $1 \leq k \leq n$.
  \qedDef
\end{definition}

Note that type constructor variables come with a parameter context.
This context determines the parameters of an initial/final dialgebra,
which essentially bundle the \emph{local} context, the domain and the
codomain of their constructors respectively destructors.

This brings us to the rules for term constructions.
To introduce them, we need further notations.
Part of the definition of the typing rules is a conversion rule,
which uses a reduction relation on types.
\begin{definition}
  The reduction relation on types consists of two types of reductions: a
  computation $\reduce$ on terms, which is defined at the end of this
  section, and $\beta$-reduction for parameters.  Essentially,
  parameter abstraction and instantiation for types corresponds to a
  simply typed $\lambda$-calculus on the type level.  Thus the
  $\beta$-reduction for parameter instantiations is given by
  \begin{equation*}
    (\pAbstr{x}{A}) \inst t \paramReduce A \subst{t}{x}.
  \end{equation*}
  The reduction relation on terms is lifted to types by taking the
  compatible closure of reduction of parameters, which is given by
  \begin{equation}
    \label{eq:ty-reduction}
    \AxiomC{$t \reduce t'$}
    \UnaryInfC{$A \inst t \reduce A \inst t'$}
    \DisplayProof
  \end{equation}
  We combine these relations into one reduction relation on types:
  $\tReduce \, = \, \paramReduce \cup \reduce$.
  One-step conversion of types is given by
  \begin{equation}
    A \tBireduce B \iff A \tReduce B \text{ or } B \tReduce A.
  \end{equation}
\end{definition}

In the typing rules for terms, we will use the following notation.
\begin{notation}
  First, we denote the identity context morphism by
  $\id_\objCtx \colonequals (x_1, \dotsc, x_n)$ for
  $\objCtx = x_1 : A_1, \dotsc, x_n : A_n$.
  Second, given a type $A$ with parameter context
  $x_1 : B_1, \dotsc, x_n : B_n$ and a context
  morphism $\sigma = (t_1, \dotsc, t_n)$, we denote by $A \inst \sigma$
  the instantiation $A \inst t_1 \inst \dotsm \inst t_n$.
\end{notation}

We continue with the term constructors.
\begin{definition}[Well-formed Terms]
  \label{def:wfterms}
  The judgement for terms is given by the rules in \iFigRef{wfterms}.
  To improve readability, we use the shorthand
  $\cramped{\rho = \recT{\rho}{X : \pType[\objCtx]}{\vec{\sigma}}{\vec{A}}}$,
  $\rho \in \{\mu, \nu\}$,
  and implicitly assume all involved types and contexts are well-formed.
  \qedDef
\begin{figure*}
  \centering
\begin{empheq}[box=\fbox]{gather*}
  \AxiomC{}
  \RightLabel{\rulelabel{$\oneT$-I}}
  \UnaryInfC{$\typed[]{\unit}{\oneT}$}
  \bottomAlignProof
  \DisplayProof
  \qquad
  \AxiomC{$\typed[\objCtx_1]{t}{\pTerm[(x : A, \objCtx_2)]{B}}$}
  \AxiomC{$\typed[\objCtx_1]{s}{A}$}
  \RightLabel{\Inst}
  \BinaryInfC{$\typed[\objCtx_1]{t \inst s}{
      \pTerm[\objCtx_2 \subst{s}{x}]{B \subst{s}{x}}}$}
  \bottomAlignProof
  \DisplayProof
  \quad
  \AxiomC{$\typed{t}{A}$}
  \AxiomC{$A \tBireduce B$}
  \RightLabel{\rulelabel{Conv}}
  \BinaryInfC{$\typed{t}{B}$}
  \DisplayProof
  \\[7pt]
  \AxiomC{$\validType[\objCtx]{A}{\univ{i}}$}
  \RightLabel{\rulelabel{Proj}}
  \UnaryInfC{$\typed[\objCtx, x:A]{x}{A}$}
  \DisplayProof
  \qquad
  \AxiomC{$\typed[\objCtx_1]{t}{\pTerm[\objCtx_2]{A}}$}
  \AxiomC{$\validTypeN[][\objCtx_1]{B}[]{i}$}
  \RightLabel{\rulelabel{Term-Weak}}
  \BinaryInfC{$\typed[\objCtx_1, x:B]{t}{\pTerm[\objCtx_2]{A}}$}
  \DisplayProof
  \\[7pt]
  \AxiomC{$\validTypeN[][]
    {\indTy{X}{\objCtx}{\vec{\sigma}}{\vec{A}}}[\objCtx]{i}$}
  \AxiomC{$1 \leq k \leq \abs{\vec{A}}$}
  \RightLabel{\rulelabel{Ind-I}}
  \BinaryInfC{$
    \typed[]{\inMuT{\objCtx}{\vec{\sigma}}{\vec{A}}_k}{
      \pTerm[(\objCtx_k, y : A_k\tySubst{\mu}{X})]{\mu \inst \sigma_k}}$}
  \DisplayProof
  \quad
  \AxiomC{$\validTypeN[][]
    {\coindTy{X}{\objCtx}{\vec{\sigma}}{\vec{A}}}[\objCtx]{i}$}
  \AxiomC{$1 \leq k \leq \abs{\vec{A}}$}
  \RightLabel{\rulelabel{Coind-E}}
  \BinaryInfC{$\typed[]{\outNuT{\objCtx}{\vec{\sigma}}{\vec{A}}_k}{
      \pTerm[(\objCtx_k, y : \nu \inst \sigma_k)]{A_k\tySubst{\nu}{X}}}$}
  \DisplayProof
  \\[7pt]
  \AxiomC{$\validTypeN[][]{C}[\objCtx]{i}$}
  \AxiomC{$\typed[\objCtxA, \objCtx_k, y_k : A_k\tySubst{C}{X}]
    {g_k}{(C \inst \sigma_k)}$}
  \AxiomC{$\forall k = 1, \dotsc, n$}
  \RightLabel{\rulelabel{Ind-E}}
  \TrinaryInfC{$
    \typed[\objCtxA]{\recPs{\objCtx_k, y_k}{g_k}}
    {\pTerm[(\objCtx, y : \mu \inst \id_{\objCtx})]{C \inst \id_{\objCtx}}}$}
  \DisplayProof
  \\[7pt]
  \AxiomC{$\validTypeN[][]{C}[\objCtx]{i}$}
  \AxiomC{$\typed[\objCtxA, \objCtx_k, y_k : (C \inst \sigma_k)]
    {g_k}{A_k\tySubst{C}{X}}$}
  \AxiomC{$\forall k = 1, \dotsc, n$}
  \RightLabel{\rulelabel{Coind-I}}
  \TrinaryInfC{$
    \typed[\objCtxA]{\corecPs{\objCtx_k, y_k}{g_k}}
    {\pTerm[(\objCtx, y : C \inst \id_{\objCtx})]{\nu \inst \id_{\objCtx}}}$}
  \DisplayProof
\end{empheq}
  \caption{Judgements for well-formed terms}
  \label{fig:wfterms}
\end{figure*}
\end{definition}

We will often leave out the type information in
the superscript of constructors and destructors.
The domain of a constructor
$\inMuT{\objCtx}{\vec{\sigma}}{\vec{A}}_k$
is determined by $A_k$ and its codomain by the instantiation $\sigma_k$.
Dually, the domain of a destructor $\outNu_k$ is given by the instantiation  $\sigma_k$
and its codomain by $A_k$.


Finally, we come to the reduction relation.
Let us agree on the following notations.
Given a context $\cramped{\objCtx = x_1 : A_1, \dotsc, x_n : A_n}$
and a type $\validTypeN[!][\objCtx]{B}[]{i}$, we denote the
full parametrisation of $B$
by $\pAbstr{\objCtx}{B} \colonequals \pAbstr{x_1}{\dotsm \pAbstr{x_n}{B}}$,
giving us
\begin{equation*}
  \validTypeN[!][\emptyset]{\pAbstr{\objCtx}{B}}[\objCtx]{i}.
\end{equation*}
Moreover, if we are given a sequence $\vec{B}$ of such types, that is,
if $\validTypeN[!][\objCtx_i]{B_i}[]{i}$ for each $B_i$ in that sequence,
we denote by $\vec{\pAbstr{\objCtx_i}{B}}$ the sequence of types that arises
by fully abstracting each type $B_i$ separately.
Finally, we denote by $\vec{B} \tySubst{C}{X}$ the
substitution of $C$ for $X$ in each $B_i$ separately.

For $C$ and $\vec{A}$ with $\validTypeN[!][\objCtx']{C}[\objCtx]{i}$,
where $\tyCtx = X_1 : \pType[\objCtx_1], \dotsc, X_n : \pType[\objCtx_n]$, and
$\validTypeN[][\objCtx_i]{A_i}[]{i}$, we define
\begin{equation*}
  \tyFunc{C}(\vec{A}) =
  C \tySubst{\vec{\pAbstr{\objCtx_i}{A}}}{\vec{X}} \inst \id_{\objCtx}.
\end{equation*}

The definition of the reduction relations requires an action of a type constructor $C$ with free
type constructor variables on terms.
Since this action will be used like a functor, we define it
(in~\iDefRef{type-action}) in such a way that the following typing rule holds.
\begin{equation}
  \label{eq:typing-functor-from-type}
  \AxiomC{$
    \validTypeN[X: {\pType[\objCtx_1]}][\objCtx_2']{C}[\objCtx_2]{i}$}
  \AxiomC{$\typed[\objCtx_1, x : A]{t}{B}$}
  \BinaryInfC{$
    \typed[\objCtx_2', \objCtx_2, x : \tyFunc{C}(A)]
    {\tyFunc{C}(t)}{\tyFunc{C}(B)}$}
  \DisplayProof
\end{equation}
In the definition of the reduction relation we need to compose
context morphisms.
This composition is for
$\objCtx_3 \xctxArr{\sigma} \objCtx_2 \xctxArr{\tau} \objCtx_1$
and $\objCtx_1 = x_1 : A_1, \dotsc, x_n:A_n$ defined by
\begin{equation}
  \label{eq:compose-ctx-mor}
  \tau \bullet \sigma \colonequals (\tau_1[\sigma], \dotsc, \tau_n[\sigma]).
\end{equation}

We can now define the reduction relation on terms.
\begin{definition}
  The reduction relation $\reduce$ on terms is defined as compatible closure of
  the contraction relation $\contract$ given in \iFigRef{term-contraction}.
  We introduce in the definition of contraction a fresh variable $x$,
  for which we immediately substitute (either $u$ or $g_k$).
  This is necessary for the use of the action of types on terms,
  see~\eqRef{typing-functor-from-type}.
  \begin{figure*}
    \vspace*{-0.3cm}
  \begin{align*}
    \recPs{\objCtx_k, y_k}{g_k} \inst \, (\sigma_k \bullet \tau) \inst \,
    (\inMu_k \inst \tau \inst u)
    & \contract g_k
    \subst*{\tyFunc{A_k}(\recPs{\objCtx_k, y_k}{g_k}
      \inst \id_{\objCtx} \inst x)}{y_k}
    [\tau, u] \\
    \outNu_k \inst \tau \inst \,
    (\corecPs{\objCtx_k, y_k}{g_k} \inst \, (\sigma_k \bullet \tau) \inst u)
    & \contract \tyFunc{A_k}\left(
      \corecPs{\objCtx_k, y_k}{g_k} \inst \id_{\objCtx} \inst x\right)
    \subst*{g_k}{x}
    [\tau, u]
  \end{align*}
    \caption{Contraction of Terms.}
    \label{fig:term-contraction}
  \end{figure*}
\end{definition}

\begin{remark}
  On terms, the reduction relation for a destructor-corecursor pair essentially
  emulates the commutativity of the following diagram (in context $\objCtx_k$).
  The dual reduction relation for a recursor-constructor pair emulates the dual
  of this diagram.
\begin{equation*}
  \begin{tikzcd}[row sep=1.5em,column sep=3.5cm]
    C \inst \sigma_k
    \rar{\corecPs{\objCtx_k, y}{g_k} \inst \sigma_k \inst x}
    \dar[swap]{g_k}
    & \nu \inst \sigma_k
    \dar{\outNu_k \inst \id_{\objCtx_k} \inst y} \\
    A_k \tySubst{C}{X}
    \rar{\tyFunc{A_k}(\corecPs{\objCtx_k, y}{g} \inst \sigma_k \inst x)}
    & A_k \tySubst{\nu}{X}
  \end{tikzcd}
\end{equation*}
\end{remark}

This concludes the definition of our proposed calculus.
Note that there are no primitive type constructors for
$\to$-, $\Pi$- or $\exists$-types, all of these are, together with the
corresponding introduction and elimination principles, definable in the above
calculus.

\begin{remark}
  \label{rem:no-induction}
  As the alert reader might have noticed, our calculus does not have dependent
  recursion and corecursion, that is, the type $C$ in \rulelabel{Ind-E} and
  \rulelabel{Coind-I} cannot depend on elements of the corresponding recursive
  type.
  This clearly makes the calculus weaker than if we had the dependent
  version:
  In the case of inductive types we do not have an induction principle,
  c.f. \iExRef{exists-quant-encode}.
  For coinductive types, on the other hand, one cannot even formulate a
  dependent version of \rulelabel{Coind-I}, rather one would expect a
  coinduction rule that turns a bisimulation into an equality proof.
  This would imply that we have an extensional function space,
  see~\iExRef{pi-encode}.
\end{remark}


\section{Examples}
\label{sec:examples}

In this section, we illustrate the calculus given in \iSecRef{syntax} on
a variety of examples.
We begin with a few basic ones, then work our way through the encoding of
logical operators, and finish with data types that are actually recursive.
Those recursive data types are lists indexed by their length (vectors) and
their dual, partial streams indexed by their definition depth.
This last example illustrates how our dependent coinductive types
generalise existing ones.

Before we go through the examples, let us introduce a notation for sequences of
empty context morphisms.
We denote such a sequence of $k$ empty context morphisms by
\begin{equation*}
  \varepsilon_k \colonequals ((), \dotsc, ()).
\end{equation*}

In the first example, we explain the role of the basic type $\oneT$.
\begin{example}[Terminal Object]
  \label{ex:terminal-obj}
  We first note that, in principle, we can encode $\oneT$ as a coinductive
  type by $\T \colonequals \coindTy{X}{}{\varepsilon_1}{X}$:
  \begin{gather*}
    \AxiomC{$\validType[\ctx{X : \pType}{\emptyset}]
      {X}{\univ{i}}$}
    \UnaryInfC{$\validTypeN[][]{
        \coindTy{X}{}{\varepsilon_1}{X}}[]{i}$}
    \DisplayProof
  \end{gather*}
  This gives us the destructor $\outNu_1 : \pTerm[(y : \T)]{\T}$ and
  the inference
  \begin{gather*}
    \AxiomC{$\validTypeN[][]{C}[]{i}$}
    \AxiomC{$\typed[y : C]{y}{C}$}
    \RightLabel{\rulelabel{Coind-I}}
    \BinaryInfC{$\typed[]{\corecP{y}{y}}{\pTerm[(y : C)]{\T}}$}
    \DisplayProof
  \end{gather*}
  So the analogue of the categorical concept of the morphism into a final
  object is given by $!_{C} \colonequals \corecP{y}{y}$.
  Note that it is not possible to define a closed term of type $\T$
  directly, rather we only get one with the help of $\oneT$ by
  $\unit' \colonequals !_{\oneT} \inst \unit$.
  Thus the puropose of $\oneT$ is to allow the formation of closed terms.
  Now, these definitions and $\tyFunc{X}(t) = t$, see \iDefRef{type-action},
  give us the following reduction.
  \begin{align*}
    \outNu_1 \inst \unit'
    & = \outNu_1 \inst \, (\corecP{y}{y} \inst \unit) \\
    & \reduce \tyFunc{X}(\corecP{y}{y} \inst x')
    \subst{y}{x'} \subst{\unit}{y} \\
    & = (\corecP{y}{y} \inst x') \subst{y}{x'} \subst{\unit}{y} \\
    & = \corecP{y}{y} \inst \unit \\
    & = \unit'
  \end{align*}
  Thus $\unit'$ is the canonical element with no observable behaviour.
  \qedDef
\end{example}

Dual to the terminal object $\T$, we can form an initial object.
\begin{example}
  \label{ex:initial-object}
  We put $\inObj \colonequals \bot \colonequals \indTy{X}{}{\varepsilon_1}{X}$,
  dual to the definition of $\T$.
  If we define $E^\bot_C \colonequals \recP{y}{y}$,
  we get the usual elimination principle for falsum:
  \begin{gather*}
    \AxiomC{$\validTypeN[][]{C}[]{i}$}
    \UnaryInfC{$\typed[]{E^\bot_C}{\pTerm[(y : \bot)]{C}}$}
    \DisplayProof
  \end{gather*}
\end{example}

Another example of a basic type are the natural numbers.
\begin{example}
  We can define the type of natural numbers by
  \begin{equation*}
    \NatT \colonequals \indTy{X}{}{\varepsilon_2}{(\T, X)},
  \end{equation*}
  with contexts $\objCtx = \objCtx_1 = \objCtx_2 = \emptyset$.
  We get the usual constructors:
  \begin{equation*}
    0 = \inMu_1^{\NatT} \inst \unit : \NatT
    \quad \text{ and } \quad
    s = \inMu_2^{\NatT} : (y : \NatT) \ctxTo \NatT.
  \end{equation*}
  Moreover, we obtain the usual recursion principle:
  \begin{equation*}
    \AxiomC{$\typed[]{t_0}{C}$}
    \AxiomC{$\typed[y : C]{t_s}{C}$}
    \BinaryInfC{$\typed[]{\recPP{}{t_0}{y}{t_s}}{\pTerm[(y : \NatT)]{C}}$}
    \DisplayProof
  \end{equation*}
\end{example}

Let us now move to logical operators.
\begin{example}[Binary Product and Coproduct]
  Suppose we are given types $\validTypeN[][\objCtx]{A_1, A_2}[]{i}$,
  then their binary product is fully specified by the two projections and
  pairing.
  Thus, we can use the following coinductive type for
  $A_1 \times_{\objCtx} A_2$.
  \begin{gather*}
    \AxiomC{$\validTypeN[][\objCtx]{A_1}[]{i}$}
    \AxiomC{$\validTypeN[][\objCtx]{A_2}[]{i}$}
    \BinaryInfC{$\validTypeN[][\Gamma]{
        \coindTy{X}{\objCtx}{(\id_\objCtx, \id_\objCtx)}{(A_1, A_2)}
        \inst \id_{\Gamma}
      }[]{i}$}
    \DisplayProof
  \end{gather*}
  Let us abbreviate
  $P \colonequals \corecPP{\objCtx, \_}{t_1}{\objCtx, \_}{t_2}$,
  then the projections are then given by
  $\pi_k \colonequals \outNu_k \inst \id_\objCtx$, and
  pairing by
  $(t_1, t_2) \colonequals P \inst \id_{\objCtx} \inst \unit$.
  We have
  \begin{equation*}
    \AxiomC{$\validTypeN[][]{\pAbstr{\objCtx}{\oneT}}[\objCtx]{i}$}
    \AxiomC{$\typed[\objCtx]{t_k}{A_k}$}
    \UnaryInfC{$\typed[\objCtx, \_ : \oneT]{t_k}{A_k}$}
    \BinaryInfC{$
      \typed[]{P}
      {\pTerm[(\objCtx, \_ : \oneT)]{A_1 \times_\objCtx A_2}}$}
    \UnaryInfC{$\typed[\objCtx]{(t_1, t_2)}{A_1 \times_\objCtx A_2}$}
    \DisplayProof
  \end{equation*}
  This setup will give us the expected reduction:
  \begin{align*}
    \pi_k \inst \, (t_1, t_2)
    & = \outNu_k \inst \id_\objCtx
    \inst \, (P \inst \id_{\objCtx} \inst \unit) \\
    & \reduce \tyFunc{A_k}(P \inst \id_{\objCtx} \inst x)
    \subst{t_k}{x} \substN{(\id_\objCtx, \unit)} \\
    & = x \subst{t_k}{x} \substN{(\id_\objCtx, \unit)} \\
    & = t_k,
  \end{align*}
  where the third step is given by $\tyFunc{A_k} = x$,
  since $A_k$ does not use type constructor variables,
  see~\iDefRef{type-action}.

  Dually, the binary coproduct of $A_1$ and $A_2$ is given by
  \begin{equation*}
    A_1 +_{\objCtx} A_2 \colonequals
    \indTy{X}{\objCtx}{(\id_\objCtx, \id_\objCtx)}{(A_1, A_2)}
    \inst \id_{\Gamma},
  \end{equation*}
  the corresponding injections by
  $\kappa_i \colonequals \inMu_i \inst \id_\objCtx$,
  and we can form the case distinction
  $\coprodArr{t_1, t_2} \colonequals \recPP{\objCtx, x}{t_1}{\objCtx, x}{t_2}$
  subject to the following typing rule.
  \begin{equation*}
    \AxiomC{$\validTypeN[][\objCtx]{C}[]{i}$}
    \AxiomC{$\typed[\objCtx, x : A_k]{t_k}{C}$}
    \BinaryInfC{$
      \typed[\objCtx]{\coprodArr{t_1,t_2}}
      {\pTerm[(x : A_1 +_\objCtx A_2)]{C}}$}
    \DisplayProof
  \end{equation*}
  Moreover, we get the expected reduction:
  \begin{equation*}
    \coprodArr{t_1,t_2} \inst \, (\kappa_i \inst s)
    \reduce t_i \subst{s}{x}.
  \end{equation*}
\end{example}

We now show that the function space arises as coinductive type.
\begin{example}[Dependent Product, $\Pi$-types, Function Space]
  \label{ex:pi-encode}
  We use $\emptyset$ as global context and $\objCtx_1 = x : A$ as local
  context, thus
  \begin{gather*}
    \AxiomC{$\validTypeN[\emptyset][x : A]{B}[]{i}$}
    \UnaryInfC{$\validTypeN[\_ : {\pType[]}][x : A]{B}[]{i}$}
    \AxiomC{$\ctxMor{\varepsilon_1}{\objCtx_1}{\emptyset}$}
    \RightLabel{\rulelabel{FP-Ty}}
    \BinaryInfC{$
      \validTypeN[][]{\coindTy{\_}{}{\varepsilon_1}{B}}[]{i}$}
    \DisplayProof
  \end{gather*}
  If we define
  $\cramped{\Pi x:A. B \colonequals \coindTy{\_}{}{\varepsilon_1}{B}}$,
  then $\cramped{\validTypeN[][]{\Pi x:A. B}[]{i}}$.
  We get $\lambda$-abstraction by putting
  $\lambda x. g \colonequals \corecP{x, \_}{g} \inst \unit$:
  \begin{equation*}
    \AxiomC{$\typed[\objCtx, x:A]{g}{B}$}
    \AxiomC{$\validTypeN[][\objCtx,x:A]{\oneT}[]{i}$}
    \RightLabel{\TermWeak}
    \BinaryInfC{$\typed[\objCtx, x:A, \_ : \oneT]{g}{B}$}
    \RightLabel{\rulelabel{Coind-I}}
    \UnaryInfC{$\typed[\objCtx]{\corecP{x, \_}{g}}
      {\pTerm[(y : \oneT)]{\Pi x:A.B}}$}
    \RightLabel{\Inst}
    \UnaryInfC{$\typed[\objCtx]
      {\lambda x. g}
      {\Pi x:A.B}$}
    \DisplayProof
  \end{equation*}
  Application is given by
  $t \, a \colonequals \outNu_1 \inst a \inst t$.
  Since we have that $B \tySubst{\Pi x:A.B}{\_} = B$ and
  $(\Pi x : A.B) [\varepsilon_1] = \Pi x : A.B$, we can derive the usual typing
  rule for application:
  \begin{equation*}
    \AxiomC{$\typed[\objCtx]{a}{A}$}
    \AxiomC{$\typed{t}{\Pi x:A. B}$}
    \BinaryInfC{$\typed[\objCtx]{t \, a}{B[a/x]}$}
    \DisplayProof
  \end{equation*}
  In particular, we have that $(\lambda x. g) \, a$ is well-typed, and
  we can derive the usual $\beta$-reduction:
  \begin{align*}
    (\lambda x. g) \, a
    & = \outNu_1 \inst a \inst \, (\corecP{x, \_{}}{g} \inst \unit) \\
    & \reduce \tyFunc{B}(\corecP{x, \_{}}{g} \inst x')
    \subst{g}{x'} \substN{\unit/\_, a/x} \\
    & = x' \subst{g}{x'} \substN{\unit/\_, a/x} \\
    & = g[a/x],
  \end{align*}
  where we again use that $\_ \not\in \fv(B)$.
  As usual, we can derive from the dependent function space also the
  non-dependent one by
  \begin{equation*}
     A \to B \colonequals \Pi x:A. B \text{ if } x \not\in \fv(B).
  \end{equation*}

  We can now extend the usual correspondence between variables in context and
  terms of function type to parameters as follows.
  First, from $\typed[\objCtx]{r}{\pTerm[(x : A)]{B}}$, we get
  $\typed[\objCtx, x : A]{r \inst x}{B}$.
  Next, for $\typed[\objCtx, x : A]{s}{B}$ we can form $\lambda x. s$, and
  finally a term $\typed[\objCtx]{t}{A \to B}$ gives rise to
  $\typed[\objCtx, x : A]{t \, x}{B}$.
  This situation can be summarised as follows.
  \begin{equation*}
    \AxiomC{$\typed[\objCtx]{r}{\pTerm[(x : A)]{B}}$}
    \UnaryInfC{$\typed[\objCtx, x : A]{s}{B}$}
    \doubleLine \dashedLine
    \UnaryInfC{$\typed[\objCtx]{t}{\Pi x:A. B}$}
    \DisplayProof
  \end{equation*}
  Here, a single line is a downwards correspondence, and the dashed double line
  signifies a two-way correspondence.
  These correspondences allow us, for example, to give the product projections
  the function type $A_1 \times_\objCtx A_2 \to A_k$, and to write
  $\pi_k \, t$ instead of $\pi_k \inst t$.

  Finally, let us illustrate how the type system only
  allows the formation of strictly positive types.
  Suppose we would want to form the type
  $X \to B$ for some variable $X$ and type $B$.
  Recall that
  $X \to B = \cramped{\coindTy{\_}{}{\varepsilon_1}{B}}$ with
  $\ctxMor{\varepsilon_1}{(x : X)}{\emptyset}$.
  However, for this to be a valid context morphism, we would need to derive
  $\isCtx{x:X}$ (note the empty context), which is not possible according to
  \iDefRef{wffctx}.
  Hence, $X \to B$ cannot be formed.
  \qedDef
\end{example}

\begin{example}[Coproducts, Existential Quantifier]
  \label{ex:exists-quant-encode}
  Recall that we do not have dependent recursion, hence
  no induction principle.
  This means that we are not able to encode $\Sigma$-types à la Martin-Löf.
  Instead, we can define intuitionistic existential quantifiers,
  see 11.4.4 and 10.8.2 in~\cite{TroelstraDalen-ConstrMath}.
  In fact, $\exists$-types occur as the dual of $\Pi$-types (\iExRef{pi-encode})
  as follows.

  Let $\validTypeN[][x : A]{B}[]{i}$ and put
  $\exists x:A. B \colonequals \indTy{\_}{}{\varepsilon_1}{B}$.
  The pairing of terms $t$ and $s$ is given by
  $(t, s) \colonequals \inMu_1 \inst t \inst s$.
  One can easily derive that $\validTypeN[][]{\exists x:A. B}[]{i}$ and
  \begin{equation*}
    \AxiomC{$\typed[\objCtx]{t}{A}$}
    \AxiomC{$\typed[\objCtx]{s}{B \subst{t}{x}}$}
    \BinaryInfC{$\typed[\objCtx]{(t,s)}{\exists x:A. B}$}
    \DisplayProof
  \end{equation*}
  from \rulelabel{Ind-I} and \Inst.
  Equally easy is also the derivation that the elimination principle for
  existential quantifiers, defined by
  \begin{equation*}
    \cramped{E_{x,y}^\exists}(t, p) \colonequals \recP{x : A, y : B}{p} \inst t,
  \end{equation*}
  can be formed by the following rule.
  \begin{equation*}
    \AxiomC{$\validTypeN[][]{C}[]{i}$}
    \AxiomC{$\typed[\objCtx, x : A, y : B]{p}{C}$}
    \AxiomC{$\typed[\objCtx]{t}{\exists x:A. B}$}
    \TrinaryInfC{$\typed[\objCtx]{\cramped{E_{x,y}^\exists}(t, p)}{C}$}
    \DisplayProof
  \end{equation*}
  Finally, we get the usual reduction rule
  \begin{equation*}
    \cramped{E_{x,y}^\exists} ((t, s), p) \reduce p \substL{t/x, s/y}.
    \tag*{\qedDef}
  \end{equation*}
\end{example}

\begin{example}[Generalised Dependent Product and Coproduct]
  \label{ex:general-(co)product}
  From a categorical perspective, it makes sense to not just consider product
  and coproducts that bind a variable in a type but also to allow the
  restriction of terms we allow as values for this variable.
  We can achieve this by replacing $\varepsilon_1$ in \iExRef{pi-encode} and
  \iExRef{exists-quant-encode} by an arbitrary term $\typed[x : I]{f}{J}$.
  This gives us type constructors with
  \begin{equation*}
    \textstyle
    \validTypeN[][y : J]{\coprod_f A}[]{i}
    \quad \text{ and } \quad
    \validTypeN[][y : J]{\prod_f A}[]{i}
  \end{equation*}
  that are weakly adjoint $\coprod_f \dashv \reidx{f} \dashv \prod_f$,
  where $\reidx{f}$ substitutes $f$.
  Similarly, propositional equality arises as left adjoint to
  contraction $\reidx{\delta}$, where
  $\ctxMor{\delta}{(x : A)}{(x : A, y : A)}$ is
  the diagonal substitution $\delta \colonequals (x, x)$,
  c.f.~\cite[Def. 10.5.1]{Jacobs1999-CLTT}.
  \qedDef
\end{example}

The next example is a standard inductive dependent type.
\begin{example}[Vectors]
  We define vectors $\VecT A : (n : \NatT) \ctxTo \univ{0}$, which are lists
  over $A$ indexed by their length, by
  \begin{align*}
    & \VecT A \colonequals
    \indTy{X}{\objCtx}{(\sigma_1, \sigma_2)}{(\T, A \times X \inst k)} \\
    & \objCtx = n : \NatT
    \quad \text{ and } \quad \objCtx_1 = \emptyset
    \quad \text{ and } \quad \objCtx_2 = k : \NatT \\
    & \ctxMor{\sigma_1 = (0)}{\objCtx_1}{(n : \NatT)}
    \quad \text{ and } \quad
    \ctxMor{\sigma_2 = (s \inst k)}{\objCtx_2}{(n : \NatT)} \\
    & \validTypeN[X : {\pType[(n : \NatT)]}][\objCtx_1]{\T}[]{0} \\
    & \validTypeN[X : {\pType[(n : \NatT)]}][\objCtx_2]
    {A \times X \inst k}[]{0}
  \end{align*}
  This yields the usual constructors $\nil \colonequals \inMu_1 \inst \unit$
  and $\cons \colonequals \inMu_2$, which have the expected types, namely
  $\inMu_1 : \T \ctxTo \VecT A \; 0$ and
  $\inMu_2 : \pTerm[(k:\NatT, y : A \times \VecT A \; k)]{\VecT A}
  \inst \, (s \inst k)$.
  The induced recursion scheme is then also the expected one.
  \qedDef
\end{example}

The dependent coinductive types of the present calculus differ from other
calculi in that destructors can be restricted in the
domain they may be applied to.
We illustrate this by defining partially defined streams, which are the dual of
vectors.
A preliminary definition is necessary though.
\begin{example}
  \label{ex:extended-nats}
  The \emph{extended naturals}, which are to be thought of as natural numbers
  extendend with an element $\infty$, are given by the following coinductive
  type.
  \begin{equation*}
    \EN = \coindTy{X}{}{\varepsilon_1}{\oneT + X} : \univ{0}
  \end{equation*}
  On this type, we can define the successor $s_\infty : y : \EN \to \EN$
  by primitive corecursion.
\end{example}

Using the extended naturals, we can define partial streams, which are streams
that might not be fully defined.
\begin{example}
  Intuitively, we define partial streams as coinductive type indexed by the
  definition depth, and destructors that can only be applied to streams that
  are defined in at least the first position:
  \begin{lstlisting}[language=Agda,mathescape=true,columns=flexible]
codata PStr (A : Set) : (n : $\EN$) $\to$ Set where
  hd : (k : $\EN$) $\to$ PStr A ($s_\infty$ k) $\to$ A
  tl : (k : $\EN$) $\to$ PStr A ($s_\infty$ k) $\to$ PStr A k
  \end{lstlisting}
  This co-datatype translates into our language by putting
  \begin{equation*}
    \PStr A \colonequals
    \coindTy{X}{\Gamma}{(s_\infty \, k, s_\infty \, k)}{(A, X \inst k)},
  \end{equation*}
  for contexts $\objCtx = n : \EN$ and $\objCtx_1 = \objCtx_2 = k : \EN$,
  context morphisms $\ctxMor{(s_\infty \, k)}{\objCtx_i}{\objCtx}$ for
  $i=1,2$, and destructor codomains
  $\validTypeN[X : {\pType[\Gamma]}][\objCtx_1]{A}[]{i}$
  and  $\validTypeN[X : {\pType[\Gamma]}][\objCtx_1]{X \inst k}[]{i}$.
  \qedDef
\end{example}


\section{Pre-Types and Pre-Terms}
\label{sec:pre-syntax}

The pre-types and pre-terms, we introduce in this section, have two purposes:
First, they allow us to justify the simultaneous definition of the
typing judgement and the reduction relation in \iSecRef{syntax}.
Second, we use them as a tool in \iSecRef{sn} to prove strong normalisation.

Pre-types and -terms are introduced as follows.
First, we define raw terms that mix types and terms, and raw contexts
whose only purpose is to ensure that arities for instantiations match.
These raw terms can then be split into pre-types and pre-terms.

\begin{definition}[Raw Syntax]
  The raw contexts and terms are given by the follows grammars.
  \begin{align*}
    \objCtx & \colonequals \emptyset \mid \objCtx, x \quad x \in \var \\
    \tyCtx & \colonequals \emptyset \mid \tyCtx, X : \pType[\objCtx] \\
    M, N & \colonequals
    \begin{aligned}[t]
    & \oneT \mid \unit
    \mid x \in \var
    \mid M \inst N
    \mid \pAbstr{x}{M}
    \mid X \in \tyVar \\
    & \mid \inMu_k
    \mid \outNu_k
    \mid \recT{\rho}{X : \pType[\objCtx]}{\vec{\sigma}}{
      \vec{M}}, \; \rho \in \{\mu, \nu\} \\
    &
    \mid \rec^{\indTy{X}{\objCtx}{\vec{\sigma}}{\vec{M}}}
    \vec{\pAbstr{\objCtx_k, y_k}{N_k}} \\
    &
    \mid \corec^{\coindTy{X}{\objCtx}{\vec{\sigma}}{\vec{M}}}
    \vec{\pAbstr{\objCtx_k, y_k}{N_k}}
    \end{aligned}
  \end{align*}
\end{definition}

Pre-types and pre-terms are defined through two judgements
\begin{equation*}
  \validTypeN{A}[\objCtx_2]{i}
  \qquad \text{ and } \qquad
  \preTerm[\objCtx_1, x]{t}[\objCtx_2].
\end{equation*}
The rules for these judgements follow essentially those in
\iSecRef{syntax}, only that the type for terms is erased.
For that reason, we leave their definition out.
However, this definition and that of the reduction
relation below have been implemented in Agda~\cite{Basold:code}.

Recall that the contraction of terms, see~\iFigRef{term-contraction}, requires
an action of (pre-)types with free type variables on terms.
We define it so that
\eqRef{typing-functor-from-type} on page~\pageref{eq:typing-functor-from-type}
holds, see~\iLemRef{correctness-type-action}.
In fact, the definition for recursive types just follows how functors
arise from parameterised initial and final dialgebras,
see~\iDefRef{parameterised-dialg} and~\cite{Basold-DepCoindFibDialg}.
\begin{definition}[Type action]
  \label{def:type-action}
  Let $\preType[!][\objCtx']{C}[\objCtx]$ be a pre-type with
  $\tyCtx = X_1 : \pType[\objCtx_1] , \dotsc, X_n : \pType[\objCtx_n]$,
  $\vec{A}$ and $\vec{B}$ be sequences of pre-types with
  $\preType[][\objCtx_i]{A_i}[]$ and $\preType[][\objCtx_i]{B_i}[]$ for all
  $1 \leq i \leq n$.
  For a sequence $\vec{t}$ of pre-terms with
  $\preTerm[\objCtx_i, x]{t}$ for all $1 \leq i \leq n$, we define
  $\tyFunc[\emptyset]{C}(\vec{t})$ as follows.
  If $n = 0$, we simply put $\tyFunc{C}(\varepsilon) = x$.
  If $n > 0$, we define $\tyFunc[\emptyset]{C}(\vec{t})$ by
  induction in the derivation of $C$.
  \begin{align*}
    & \begin{aligned}
      & \tyFunc[\tyCtx, X_{n+1}]{C}(\vec{t}, t_{n+1})
      = \tyFunc[\tyCtx]{C}(\vec{t})
      & & \text{for } \TyVarWeak \\
      & \tyFunc[\tyCtx]{X_i}(\vec{t}) = t_i \\
      & \tyFunc[\tyCtx]{C' \inst s}(\vec{t})
      = \tyFunc[\tyCtx]{C'}(\vec{t}) \subst{s}{y},
      & & \text{for }
      \validTypeN[!][\objCtx']{C'}[(y, \objCtx)]{i} \\
      & \tyFunc[\tyCtx]{\pAbstr{y}{C'}}(\vec{t}) = \tyFunc[\tyCtx]{C'}(\vec{t}),
      & & \text{for }
      \validTypeN[!][(\objCtx',y)]{C'}[\objCtx]{i} \\
    \end{aligned} \\
    & \tyFunc[\tyCtx]{\indTy{Y}{\objCtx}{\vec{\sigma}}{\vec{D}}}(\vec{t})
      = \recPs[R_A]{\objCtxA_k, x}{g_k} \inst \id_{\objCtx} \inst x, \\
      & \quad \text{with } g_k =
      \inMu_k
      \inst \id_{\objCtxA_k} \inst
      \left(\tyFunc[\tyCtx, Y]{D_k}\left(
          \vec{t}, y
        \right)\right) \\
      & \quad \text{and } R_{A} =
      \indTy{Y}{\objCtx}{\vec{\sigma}}{\vec{D} \tySubst{\vec{\pAbstr{\objCtx_i}{A}}}{\vec{X}}} \\
      & \quad \text{for }
      \validTypeN[\tyCtx, Y : {\pType[\objCtx]}][\Delta_k]{D_k}[]{i} \\
    & \tyFunc[\tyCtx]{\coindTy{Y}{\objCtx}{\vec{\sigma}}{\vec{D}}}(\vec{t})
      = \corecPs[R_{B}]{\objCtxA_k, x}{g}
      \inst \id_{\objCtx} \inst x, \\
      & \quad \text{with } g_k =
      \tyFunc[\tyCtx,Y]{D_k}\left( \vec{t}, x \right)
      \subst*{(\outNu_k \inst \id_{\objCtxA_k} \inst x)}{x} \\
      & \quad \text{and } R_{B} =
      \coindTy{Y}{\objCtx}{\vec{\sigma}}{\vec{D} \tySubst{\vec{\pAbstr{\objCtx_i}{B}}}{\vec{X}}}
  \end{align*}

\end{definition}


\section{Meta properties}
\label{sec:meta-prop}

\subsection{Subject Reduction}
\label{sec:subject-reduction}

The proof of subject reduction is based on the following key lemma,
which essentially states that the action of types on terms acts like a
functor.
\begin{lemma}[Type correctness of type action]
  \label{lem:correctness-type-action}
  Given the action of types on terms, see \iDefRef{type-action},
  the following inference rule holds.
  \begin{equation*}
    \AxiomC{$
      \validTypeN[X: {\pType[\objCtx_1]}][\objCtx_2']{C}[\objCtx_2]{i}$}
    \AxiomC{$\typed[\objCtx_1, x : A]{t}{B}$}
    \BinaryInfC{$
      \typed[\objCtx_2', \objCtx_2, x : \tyFunc{C}(A)]
      {\tyFunc{C}(t)}{\tyFunc{C}(B)}$}
    \DisplayProof
  \end{equation*}
\end{lemma}
\begin{proof}
  We leave the proof details out.
  Let us only mention that the proof works by
  generalising the statement to types in arbitrary type constructor contexts
  $\tyCtx$, and then proceeding by induction in $C$.
\end{proof}

The following is now an easy consequence of \iLemRef{correctness-type-action}.
\begin{theorem}[Subject reduction]
  If $\cramped{\typed{t_1}{A}}$ and $t_1 \reduce t_2$,
  then $\typed{t_2}{A}$.
\end{theorem}

\subsection{Strong Normalisation}
\label{sec:sn}

This section is devoted to show that all terms $\typed{t}{A}$ are strongly
normalising, which is intuitively what one would expect, given that we
introduced the reduction relation by following the homomorphism property
of (co)recursion for initial and final dialgebras.

The proof uses the saturated sets approach, see~\cite{Geuvers:SN-CoC-SAT}, as
follows.
First, we define what it means for a set of pre-terms to be saturated, where,
most importantly, all terms in a saturated set are strongly normalising.
Next, we give an interpretation $\sem{A}$ of dependent types $A$ as families of
saturated sets.
Finally, we show that if $\typed{t}{A}$, then for all assignments $\rho$ of
terms to variables in $\objCtx$, we have $t \in \sem{A}(\rho)$.
Since $\sem{A}(\rho) \subseteq \SN$, strong normalisation for all typed terms
follows.

We begin with a few simple definitions.
\begin{definition}
  We use the following notations.
  \begin{itemize}
  \item $\Terms$ is the set of pre-terms.
  \item $\SN$ is the set of strongly normalising pre-terms.
  \item $\vars{\Gamma}$ is the set of variables in context $\Gamma$.
  \end{itemize}
\end{definition}

For simplicity, we identify context morphisms
$\ctxMor{\sigma}{\objCtx_1}{\objCtx_2}$ and valuations
$\rho : \vars{\objCtx_2} \to \Terms$, if we know that the terms of $\rho$ live
in $\objCtx_1$.
This allows us to write $\sigma(x)$ for $x \in \vars{\objCtx_2}$, and
$M \inst \rho$ for pre-terms $M$.
It is helpful though to make the action of context morphisms on
valuations, essentially given by composition, explicit and write
\begin{equation}
  \label{eq:sem-ctx-mor}
  \begin{aligned}
    & \sem{\ctxMor{\sigma}{\objCtx_1}{\objCtx_2}} :
    \Terms^{\vars{\Gamma_1}} \to \Terms^{\vars{\Gamma_2}} \\
    & \sem{\ctxMor{\sigma}{\objCtx_1}{\objCtx_2}}(\gamma)(y)
    = \sigma(y)[\gamma].
  \end{aligned}
\end{equation}

Saturated sets are defined by containing certain open terms (base terms)
and by being closed under key reductions.
We introduce these two notions in the following two definitions.
\begin{definition}[Base Terms]
  The set of \emph{base terms} $\base$ is defined inductively by the following
  three closure rules.
  \begin{itemize}
  \item $\var \subseteq \base $
  \item
    $\rec \vec{\pAbstr{\objCtx_k, x}{N_k}} \inst \sigma \inst M \in \base$,
    provided that $M \in \base$, $N_k \in \SN$ and $\sigma \in \SN$.
  \item
    $\outNu_k^{\coindTy{X}{\objCtx}{\vec{\tau}}{\vec{A}}} \inst \sigma \inst M
    \in \base$,
    provided that $M \in \base$, $\sigma \in \SN$ and
    $\exists \gamma. \left(\sigma = \sem{\tau_k}(\gamma) \right)$.
  \end{itemize}
\end{definition}

\begin{definition}[Key Redex]
  A pre-term $M$ is a \emph{redex}, if there is a $P$ with $M \contract P$.
  $M$ is the \emph{key redex}
  \begin{enumerate}
  \item of $M$ itself, if $M$ is a redex,
  \item of $\rec \vec{\pAbstr{\objCtx_k, y_k}{N_k}} \inst \sigma \inst N$,
    if $M$ the key redex of $N$, or
  \item of $\outNu_k \inst \sigma \inst N$,
    if $M$ the key redex of $N$.
  \end{enumerate}
  We denote by $\kred{M}$ the term that is obtained by contracting the key
  redex of $M$.
\end{definition}

\begin{definition}[Saturated Sets]
  A set $X \subseteq \Terms$ is \emph{saturated}, if
  \begin{enumerate}
  \item $X \subseteq \SN$
  \item $\base \subseteq X$
  \item If $\kred{M} \in X$ and $M \in \SN$, then $M \in X$.
  \end{enumerate}
  We denote by $\sat$ the set of all saturated sets.
\end{definition}

It is easy to see that $\SN \in \sat$, and that every saturated set is
non-empty. Moreover, it is easy to show that $\sat$ is a complete
lattice with set inclusion as order.
Besides these standard facts, we will use the following constructions
on saturated sets.

\begin{definition}
  Let $\objCtx$ be a context.
  We define a form of semantical context extension (comprehension)
  of pairs $(E, U)$ with $E \subseteq \Terms^{\vars{\objCtx}}$ and
  $U : E \to \sat$ with respect to a given variable $x \not\in \vars{\objCtx}$
  by
  \begin{equation}
    \label{eq:sat-comprehension}
    \compr{(E, U)}_{x} =
    \setDef*{\rho[x \mapsto M]}{
      \rho \in E \text{ and }
      M \in U(\rho)},
  \end{equation}
  where $\rho[x \mapsto M] : \vars{\objCtx} \cup \{x\} \to \Terms$
  extends $\rho$ by mapping $x$ to $M$.
  Moreover, we define a semantical version of the typing judgement:
  \begin{equation}
    \label{eq:sat-typing}
    E \Vdash U = \setDef{M}{\forall \gamma \in E. \; M[\gamma] \in U(\gamma)}.
  \end{equation}
\end{definition}

We now show that we can give a model of well-formed types by means of saturated
sets.
To achieve this, we define simultaneously an interpretation of contexts and the
interpretation of types.
The intention is that we have that
\begin{itemize}
\item if $\isCtx{\objCtx}$, then
  $\sem{\objCtx} \subseteq \Terms^{\vars{\objCtx}}$,
\item if $\isTyCtx{\tyCtx}$, then
  $\sem{\tyCtx}(X) \in \sat^{\sem{\objCtx}}$ for
  all $X : \pType[\objCtx]$ in $\tyCtx$, and
\item if $\validTypeN{A}{i}$, then
  $\sem{A} : \sem{\tyCtx} \times \sem{\objCtx_1, \objCtx_2} \to \sat$.
\end{itemize}

\begin{definition}[Interpretations]
  \label{def:interpretation}
  We interpret type variable contexts, term variable contexts and
  types simultaneously.
  First, we assign to each term context is a set of allowed possible
  valuations:
  \begin{align*}
    \sem{\emptyset} & = \{! : \emptyset \to \Terms\} \\
    \sem{\objCtx, x : A} & =
    \compr{(\sem{\objCtx}, \sem{A})}_x \\
    & = \setDef{\rho[x \mapsto M]}{
      \rho \in \sem{\objCtx} \text{ and } M \in \sem{A}(\rho)}
  \end{align*}
  For $\tyCtx = X_1 : \pType[\objCtx_1], \dotsc, X_n : \pType[\objCtx_n]$ we
  define
  \begin{align*}
    \sem{\tyCtx} = \prod_{X_i \in \vars{\tyCtx}} \tyVal{\objCtx_i},
  \end{align*}
  where $I_\objCtx$ is the set of valuations that respect convertibility:
  \begin{equation*}
    \tyVal{\objCtx} =
    \setDef{U : \sem{\objCtx} \to \sat}{
      \forall \rho, \rho'. \; \rho \tReduce \rho' \Rightarrow U(\rho) = U(\rho')}
  \end{equation*}
  Finally, we define in \iFigRef{type-interpretation} the interpretation of
  types as families of term sets.
  \begin{figure*}
  \begin{align*}
    \sem{\validTypeN[][]{\oneT}[]{i}}(\delta, \rho)
    & = \bigcap \setDef{X \in \sat}{\unit \in X} \\
    \sem{\validTypeN[\tyCtx, X : {\pType[\objCtx]}][\emptyset]{X}[\objCtx]{i}}
    (\delta, \rho)
    & = \delta(X)(\rho) \\
    \sem{\validTypeN[\tyCtx, X]{A}{i}}(\delta, \rho)
    & = \sem{\validTypeN{A}{i}}(\restr{\delta}{\vars{\tyCtx}}, \rho) \\
    \sem{\validTypeN[!][\objCtx_1, x : B]{A}{i}}(\delta, \rho)
    & = \sem{\validTypeN{A}{i}}(\delta, \restr{\rho}{\vars{\objCtx_1}}) \\
    \sem{\validTypeN{A \inst t}[\objCtx_2 \subst{t}{x}]{i}}(\delta, \rho)
    & = \sem{\validTypeN{A}[(x : B, \objCtx_2)]{i}}(\delta, \rho[x \mapsto t[\rho]]) \\
    \sem{\validTypeN{\pAbstr{x}{A}}[(x : B, \objCtx_2)]{i}}(\delta, \rho)
    & = \sem{\validTypeN[!][\objCtx_1, x : B]{A}{i}}(\delta, \rho) \\
    \sem{\validTypeN[!][\emptyset]{
        \indTy{X}{\objCtx}{\vec{\sigma}}{\vec{A}}}[\objCtx]{i}}(\delta, \rho)
    & = 
    \begin{aligned}[t]
      \{M \mid \;
      & \forall U \in \tyVal{\objCtx}. \; \forall \objCtxA. \;
      \forall k. \;
      \forall N_k \in
      \compr{\sem{\objCtx_k}, \sem{A_k^\objCtxA}(\delta[X \mapsto U])}_y
      \Vdash \reidx{\sem{\sigma_k \bullet \pi}}(U). \\
      & \quad
      \rec \vec{\pAbstr{\objCtx_k, y}{N_k}} \inst \rho \inst M \in U(\rho) \}
    \end{aligned} \\
    \sem{\validTypeN[!][\emptyset]{
        \coindTy{X}{\objCtx}{\vec{\sigma}}{\vec{A}}}[\objCtx]{i}}(\delta, \rho)
    & = 
      \{M \mid \;
      \exists U \in \tyVal{\objCtx}. \;
      \forall k. \;
      \forall \gamma \in \sem{\sigma_k}^{-1}(\rho). \;
      \outNu_k \inst \gamma \inst M \in \sem{A_k}(\delta[X \mapsto U], \gamma)\}
  \end{align*}
    \caption{Interpretation of types as families of saturated sets}
    \label{fig:type-interpretation}
  \end{figure*}
  In the clause for inductive types, $A_k^\objCtxA$ denotes the type that
  is obtained by weakening $\validTypeN[][\objCtx_k]{A_k}[]{i}$ to
  $\validTypeN[][\objCtxA, \objCtx_k]{A_k^\objCtxA}[]{i}$,
  $\ctxMor{\pi}{(\objCtx_k, y : A_k^\Delta)}{\objCtx_k}$ projects $y$ away, and
  $\reidx{\sem{\sigma_k \bullet \pi}}(U) \colonequals
  U \circ \sem{\sigma_k \bullet \pi}$
  is the reindexing for set families.
  \qedDef
\end{definition}

Before we continue stating the key results about this interpretation of types,
let us briefly look at an example.
\begin{example}
  Suppose $A, B$ are closed types.
  Recall that the function space was
  $A \to B = \coindTy{X}{}{\ctxMor{\varepsilon_1}{(x:A)}{\emptyset}}{B}$,
  and that application was defined by $t \, a = \outNu_1 \inst a \inst t$.
  Note that the conditition $\gamma \in \sem{\varepsilon_1}^{-1}(\rho)$ reduces
  to $\gamma(x) \in \sem{A}$ because
  $\sem{\varepsilon_1}(\gamma)(y \in \emptyset) = \rho(y)[\varepsilon_1]$ holds
  for any $\gamma \in \sem{(x : A)}$.
  So we write $N$ instead of $\gamma(x)$.
  We further note that, since $A$, $B$ and thus $A \to B$ are closed,
  we can leave out the valuation $\delta$ for the type variables.
  Taking all of this into account, we have
  \begin{align*}
    \sem{A \to B}(\gamma)
    & = \setDef{M}{\forall N \in \sem{A}. \,
      \outNu_1 \inst \gamma \inst M \in \sem{B}} \\
    & = \setDef{M}{\forall N \in \sem{A}. \, M \, N \in \sem{B}},
  \end{align*}
  which is the usual definition definition, see~\cite{Geuvers:SN-CoC-SAT}.
  \qedDef
\end{example}

\begin{remark}
  One interesting result, used to prove the following
  lemmas, is that the interpretation of types is monotone in $\delta$,
  and that the interpretation of coinductive types is the largest set
  closed under destructors.  This suggests that it might be possible
  to formulate the definition of the interpretation in categorical
  terms.
\end{remark}

We just state here the key lemmas and leave their proofs out.
\begin{lemma}[Soundness of type action]
  \label{lem:ty-functor-sound}
  Suppose $C$ is a type with $\validTypeN[\tyCtx][\objCtx]{C}[\objCtx']{i}$
  and $\tyCtx = X_1 : \pType[\objCtx_1], \dotsc, X_n : \pType[\objCtx_n]$,
  such that for all parameters $\typed[\objCtxA]{r}{C'}$ occurring in $C$
  and $\ctxMor{\tau}{\objCtxA'}{\objCtxA}$, we have
  $r[\tau] \in \sem{C'}(\tau)$.
  Let $\delta_A, \delta_B \in \sem{\tyCtx}$
  and $\typed[\objCtx_i, x : A_i]{t_i}{B_i}$, such that for all
  $\sigma \in \sem{\objCtx_i}$,
  $t_i \in \delta_B(X_i)(\tau)$.
  Then for all contexts $\objCtxA$,
  all $\ctxMor{\sigma}{\objCtxA}{\objCtx, \objCtx'}$ and all
  $s \in \sem{C}(\delta_A, \sigma)$
  \begin{equation}
    \tyFunc{C}(\vec{t})\substN{\sigma, s} \in \sem{C}(\delta_B, \sigma).
  \end{equation}
\end{lemma}

\begin{lemma}
  \label{lem:interpret-sat}
  The interpretation of types $\sem{-}$ given in \iDefRef{interpretation}
  is well-defined and
  $\sem{A}(\delta, \rho) \in \sat$ for all $A, \delta, \rho$.
\end{lemma}

\begin{lemma}[Soundness]
  \label{lem:sat-sound}
  If $\typed[\objCtx]{t}{A}$, then for all $\rho \in \sem{\objCtx}$ we have
  $t[\rho] \in \sem{A}(\rho)$.
\end{lemma}

From the soundness, we can easily derive strong normalisation.
\begin{theorem}
  All well-typed terms are strongly normalising, that is,
  if $\typed[\objCtx_1]{t}{\pTerm[\objCtx_2]{A}}$ then
  $t \in \SN$.
\end{theorem}
\begin{proof}
  We first note that terms only reduce if $\objCtx_2 = \emptyset$.
  In that case we can apply we can apply \iLemRef{sat-sound} with
  $\rho$ being the identity, so that $t \in \sem{A}(\rho)$.
  Thus, by \iLemRef{interpret-sat} and the definition of saturated sets,
  we can conclude that $t \in \SN$.
  Since $t$ does not reduce if $\objCtx_2$ is non-trivial, we also have
  in that case that $t \in \SN$.
  Hence every well-typed term is strongly normalising.
\end{proof}


\section{Conclusion}
\label{sec:conclusion}

We have introduced a type theory that is solely based on inductive and
coinductive types, in contrast to other type theories that usually have
separate type constructors for, for example, the function space.
This results in a a small set of rules for the judgements of the theory
and the corresponding reduction relation.
To justify the use of our type theory as logic, we also proved that the
reduction relation preserves types and is strongly normalising on well-typed
terms.
Combining the present theory with that in~\cite{norell:thesis}
would give us a justification for a large part Agda's current type system,
especially including coinductive types.

There are still some open questions, regarding the present type theory,
that we wish to settle in the future.
First of all, a basic property of the reduction relation that is still missing
is confluence.
Second, we have constructed the type theory with certain categorical
structures in mind, and it is easy to interpret the types and terms in a data
type closed category, see~\cite{Basold-DepCoindFibDialg}.
However, one has to be careful in showing the soundness of such an
interpretation, that is, if $A \tBireduce B$ then we better have
$\sem{A} = \sem{B}$ because of the conversion rule.
We can indeed define a generalised Beck-Chevalley condition,
c.f.~\cite{Jacobs1999-CLTT}, but it remains to be checked if this condition is
strong enough to ensure soundness.

In \remRef{no-induction}, we mentioned already that we have no dependent
recursion, i.e., no induction.
This could be added to the theory, using a lifting of the types $A_k$ in
$\indTy{X}{\objCtx}{\vec{\sigma}}{\vec{A}}$ to predicates, and we think that
the proof of strong normalisation can be adopted accordingly.
We did not develop this in the present paper in order to keep matters simple for
the time being.

Moreover, we would like to investigate at the same time the principle dual to
induction, namely proving an equality by means of a bisimulation.
It is, however, not clear how this can be properly integrated.
There are several options, all of which are not completely satisfactory:
Equip all types with bisimilarity $\sim$ as specific equality (setoid approach),
which makes the theory hard to use; add a generalised replacement rule that,
given a proof $p : t \sim s$, allows us to infer $r : P \inst s$ from
$r : P \inst t$, but this makes type checking undecidable; add a cast operator
like in observational type theory, but this introduces non-canonical
terms~\cite{Altenkirch2007-ObsEq}; treat coinductive
types similar to higher inductive types~\cite{hottbook}, but it is not clear how
such a system can be set up, though it seems to be the most promising approach.

Finally, certain acceptable facts are not provable in our theory, since we do
not have universes.
Another common feature that is missing from the type theory, is predicative
polymorphism, which is in fact justified and desirable from the categorical
point of view.


\bibliographystyle{abbrvnat}
\bibliography{FibDialg}

\begin{thebibliography}{29}
\providecommand{\natexlab}[1]{#1}
\providecommand{\url}[1]{\texttt{#1}}
\expandafter\ifx\csname urlstyle\endcsname\relax
  \providecommand{\doi}[1]{doi: #1}\else
  \providecommand{\doi}{doi: \begingroup \urlstyle{rm}\Url}\fi

\bibitem[{Abel} and {Pientka}(2013)]{Abel2013}
A.~{Abel} and B.~{Pientka}.
\newblock Wellfounded recursion with copatterns: a unified approach to
  termination and productivity.
\newblock In \emph{{{ICFP}}}, pages 185--196, 2013.
\newblock \doi{10.1145/2500365.2500591}.

\bibitem[{Abel} et~al.(2013){Abel}, {Pientka}, {Thibodeau}, and
  {Setzer}]{Abel2013-Copatterns}
A.~{Abel}, B.~{Pientka}, D.~{Thibodeau}, and A.~{Setzer}.
\newblock Copatterns: {{Programming Infinite Structures}} by {{Observations}}.
\newblock In \emph{Proc. of {{POPL}}}, pages 27--38. {ACM}, 2013.
\newblock \doi{10.1145/2429069.2429075}.

\bibitem[{Aczel}(1988)]{Aczel:nonwfs}
P.~{Aczel}.
\newblock \emph{Non-well-founded {{Sets}}}.
\newblock Number~14 in Lecture Notes. {Center for the Study of Language and
  Information, Stanford University}, 1988.

\bibitem[{Agda}(2015)]{Agda:system}
{Agda}.
\newblock Agda {{Documentation}}.
\newblock Technical report, Programming Logic group, Chalmers and Gothenburg
  University, 2015.
\newblock URL \url{http://wiki.portal.chalmers.se/agda/}.
\newblock Version 2.4.2.5.

\bibitem[{Ahrens} et~al.(2015){Ahrens}, {Capriotti}, and
  {Spadotti}]{Ahrens:MTypes-HoTT}
B.~{Ahrens}, P.~{Capriotti}, and R.~{Spadotti}.
\newblock Non-wellfounded trees in {{Homotopy Type Theory}}.
\newblock \emph{ArXiv150402949 Cs Math}, Apr. 2015.

\bibitem[{Altenkirch} et~al.(2007){Altenkirch}, {McBride}, and
  {Swierstra}]{Altenkirch2007-ObsEq}
T.~{Altenkirch}, C.~{McBride}, and W.~{Swierstra}.
\newblock Observational {{Equality}}, {{Now}}!
\newblock In \emph{Proc. of {{PLPV}} '07}, pages 57--68. {ACM}, 2007.
\newblock \doi{10.1145/1292597.1292608}.

\bibitem[{Basold}(2015)]{Basold-DepCoindFibDialg}
H.~{Basold}.
\newblock Dependent {{Inductive}} and {{Coinductive Types}} are {{Fibrational
  Dialgebras}}.
\newblock In R.~{Matthes} and M.~{Mio}, editors, \emph{Proceedings of {{FICS}}
  '15}, volume 191 of \emph{EPTCS}, pages 3--17. {Open Publishing Association},
  Sept. 2015.
\newblock \doi{10.4204/EPTCS.191.3}.

\bibitem[{Basold}(2016)]{Basold:code}
H.~{Basold}.
\newblock {Code Repository}, 2016.
\newblock URL \url{http://cs.ru.nl/~hbasold/code/}.

\bibitem[{Basold} and {Hansen}(2015)]{BasoldHansen:Welldef-Equiv-CoInd}
H.~{Basold} and H.~H. {Hansen}.
\newblock Well-definedness and {{Observational Equivalence}} for
  {{Inductive-Coinductive Programs}}.
\newblock \emph{To Appear in JLC}, 2015.

\bibitem[{Bertot} and {Cast{\'e}ran}(2004)]{BertotC04}
Y.~{Bertot} and P.~{Cast{\'e}ran}.
\newblock \emph{Interactive {{Theorem Proving}} and {{Program Development}} -
  {{Coq}}'{{Art}}: {{The Calculus}} of {{Inductive Constructions}}}.
\newblock Texts in Theoretical Computer Science. An EATCS Series. {Springer},
  2004.

\bibitem[{Bizjak} et~al.(2016){Bizjak}, {Clouston}, {Grathwohl},
  {M{\o}gelberg}, and {Birkedal}]{Bizjak:GuardedDepRecTypes}
A.~{Bizjak}, R.~{Clouston}, H.~B. {Grathwohl}, R.~E. {M{\o}gelberg}, and
  L.~{Birkedal}.
\newblock Guarded {{Dependent Type Theory}} with {{Coinductive Types}}.
\newblock In \emph{Proceedings of {{FOSSACS}}'16}, 2016.

\bibitem[{cLab}(2016)]{cLab:FinalChain-TT}
{cLab}.
\newblock Type {{Theoretic Interpretation}} of the {{Final Chain}}, 2016.
\newblock URL
  \url{https://coalg.org/clab/Type_Theoretic_Interpretation_of_the_Final_Chain}.

\bibitem[Constable(1997)]{Constable:NuprlTT}
R.~L. Constable.
\newblock The {{Structure}} of {{Nuprl}}'s {{Type Theory}}.
\newblock In \emph{Logic of {{Computation}}}, number 157 in NATO ASI Series,
  pages 123--155. {Springer Berlin Heidelberg}, 1997.

\bibitem[{Coq Development Team}(2012)]{Coq:manual}
{Coq Development Team}.
\newblock The {{Coq}} proof assistant reference manual.
\newblock Technical report, LogiCal Project, 2012.
\newblock URL \url{http://coq.inria.fr}.
\newblock Version 8.4.

\bibitem[{Geuvers}(1994)]{Geuvers:SN-CoC-SAT}
H.~{Geuvers}.
\newblock A short and flexible proof of strong normalization for the calculus
  of constructions.
\newblock In P.~{Dybjer}, B.~{Nordstr{\"o}m}, and J.~{Smith}, editors,
  \emph{Types for {{Proofs}} and {{Programs}}}, number 996 in LNCS, pages
  14--38. {Springer Berlin Heidelberg}, June 1994.
\newblock \doi{10.1007/3-540-60579-7_2}.

\bibitem[{Gim{\'e}nez}(1995)]{Gimenez-RecursiveSchemes}
E.~{Gim{\'e}nez}.
\newblock Codifying {{Guarded Definitions}} with {{Recursive Schemes}}.
\newblock In \emph{Selected {{Papers}} from the {{TYPES}} '94 {{Workshop}}},
  pages 39--59, London, UK, 1995. {Springer-Verlag}.

\bibitem[{Hagino}(1987)]{Hagino-Dialg}
T.~{Hagino}.
\newblock A typed lambda calculus with categorical type constructors.
\newblock In \emph{Category {{Theory}} in {{Computer Science}}}, pages
  140--157, 1987.

\bibitem[{Hamana} and {Fiore}(2011)]{HamanaFiore-GADT}
M.~{Hamana} and M.~{Fiore}.
\newblock A {{Foundation}} for {{GADTs}} and {{Inductive Families}}:
  {{Dependent Polynomial Functor Approach}}.
\newblock In \emph{Proceedings of the {{Seventh WGP}}}, WGP '11, pages 59--70,
  New York, NY, USA, 2011. {ACM}.
\newblock \doi{10.1145/2036918.2036927}.

\bibitem[{Jacobs}(1999)]{Jacobs1999-CLTT}
B.~{Jacobs}.
\newblock \emph{Categorical {{Logic}} and {{Type Theory}}}.
\newblock Number 141 in Studies in Logic and the Foundations of Mathematics.
  {North Holland}, Amsterdam, 1999.

\bibitem[{Lambek} and {Scott}(1988)]{LambekScott-HOCatLog}
J.~{Lambek} and P.~J. {Scott}.
\newblock \emph{Introduction to {{Higher-Order Categorical Logic}}}.
\newblock {Cambridge University Press}, Mar. 1988.

\bibitem[{Martin-L{\"o}f}(1975)]{MartinLof-ModelsITT}
P.~{Martin-L{\"o}f}.
\newblock About {{Models}} for {{Intuitionistic Type Theories}} and the
  {{Notion}} of {{Definitional Equality}}.
\newblock In \emph{3rd {{Scandinavian Logic Symposium}}}, pages 81--109. {North
  Holland and American Elsevier}, 1975.

\bibitem[{Norell}(2007)]{norell:thesis}
U.~{Norell}.
\newblock \emph{Towards a practical programming language based on dependent
  type theory}.
\newblock PhD thesis, Chalmers University of Technology, G{\"o}teborg, Sweden,
  Sept. 2007.

\bibitem[Paulin{-}Mohring(1993)]{Paulin-Mohring93}
C.~Paulin{-}Mohring.
\newblock Inductive definitions in the system {C}oq - rules and properties.
\newblock In M.~Bezem and J.~F. Groote, editors, \emph{International Conference
  on Typed Lambda Calculi and Applications, {TLCA}, Proceedings}, volume 664 of
  \emph{LNCS}, pages 328--345. Springer, 1993.

\bibitem[{Sacchini}(2013)]{Sacchini-TypeBasedProductivity}
J.~{Sacchini}.
\newblock Type-{{Based Productivity}} of {{Stream Definitions}} in the
  {{Calculus}} of {{Constructions}}.
\newblock In \emph{2013 28th {{Annual IEEE}}/{{ACM Symposium}} on {{Logic}} in
  {{Computer Science}} ({{LICS}})}, pages 233--242, June 2013.
\newblock \doi{10.1109/LICS.2013.29}.

\bibitem[{The Univalent Foundations Program}(2013)]{hottbook}
{The Univalent Foundations Program}.
\newblock \emph{Homotopy {{Type Theory}}: {{Univalent Foundations}} of
  {{Mathematics}}}.
\newblock {http://homotopytypetheory.org/book}, Institute for Advanced Study,
  2013.

\bibitem[{Troelstra} and {van Dalen}(1988)]{TroelstraDalen-ConstrMath}
A.~S. {Troelstra} and D.~{van Dalen}.
\newblock \emph{Constructivism in {{Mathematics}}: {{An Introduction}}}.
\newblock {North-Holland}, 1988.

\bibitem[{van den Berg}(2006)]{vanDenBerg:thesis}
B.~{van den Berg}.
\newblock \emph{Predicative topos theory and models for constructive set
  theory}.
\newblock Phd, University of Utrecht, 2006.

\bibitem[{van den Berg} and {De Marchi}(2007)]{vdBerg-Non-wellfoundedTrees}
B.~{van den Berg} and F.~{De Marchi}.
\newblock Non-well-founded trees in categories.
\newblock \emph{Annals of Pure and Applied Logic}, 146\penalty0 (1):\penalty0
  40--59, Apr. 2007.
\newblock ISSN 0168-0072.
\newblock \doi{10.1016/j.apal.2006.12.001}.

\bibitem[{Werner}(1994)]{Werner1994}
B.~{Werner}.
\newblock \emph{Une th{\'e}orie des {{Constructions Inductives}}}.
\newblock Phd, Universit{\'e} Paris VII, 1994.

\end{thebibliography}

\clearpage
\appendix
\section{Examples}
\label{app:examples}

\begin{example}
  In \iExRef{initial-object}, we have defined the bottom type (initial object)
  to be $\bot = \inObj = \indTy{X}{}{\varepsilon_1}{X}$.
  Note that this gave us the elimination principle
  \begin{gather*}
    \AxiomC{$\validTypeN[][]{C}[]{i}$}
    \UnaryInfC{$\typed[]{E^\bot_C}{\pTerm[(y : \bot)]{C}}$}
    \DisplayProof
  \end{gather*}
  in which $C$ is not allowed to have any dependencies.
  So, to make $\bot$ usable with dependent types, we need to switch to
  fibred initial objects:
  \begin{equation*}
    \bot_{\objCtx} \colonequals
    \indTy{X}{\objCtx}{\id_{\objCtx}}{X \inst \id_{\objCtx}} \inst \id_{\objCtx},
  \end{equation*}
  which corresponds to the following generalised algebraic data type.
  \begin{lstlisting}[language=Agda,mathescape=true,columns=flexible]
data $\bot_{\objCtx}$ : $(x_1 : A_1)$ $\dotsm$ $(x_n : A_n)$ $\to$ Set where
  $\inMu$ : $(x_1 : A_1)$ $\dotsm$ $(x_n : A_n)$ $\to$ $\bot_{\objCtx}$ $x_1$ $\dotsm$ $x_n$ $\to$ $\bot_{\objCtx}$ $x_1$ $\dotsm$ $x_n$
  \end{lstlisting}
  Similar to before, we define
  $E^\bot_C \colonequals \recP{\objCtx, y}{y} \inst \id_\objCtx$,
  to obtain the elimination principle for falsum in context $\objCtx$:
  \begin{gather*}
    \AxiomC{$\validTypeN[][\objCtx]{C}[]{i}$}
    \UnaryInfC{$\typed[\objCtx]{E^\bot_C}{\pTerm[(x : \bot_{\objCtx})]{C}}$}
    \DisplayProof
  \end{gather*}
  The corresponding derivation is given as follows.
  \begin{gather*}
    \AxiomC{$\validTypeN[][\objCtx]{C}[]{i}$}
    \UnaryInfC{$\validTypeN[][]{\pAbstr{\objCtx}{C}}[\objCtx]{i}$}
    \AxiomC{$\typed[\objCtx, y : C]{y}{C}$}
    \RightLabel{\IndE}
    \BinaryInfC{$\typed[]{\recP{\objCtx, y}{y}}{
        \pTerm[(\objCtx, y : \bot_\objCtx)]
        {(\pAbstr{\objCtx}{C}) \inst \id_\objCtx}}$}
    \RightLabel{\Inst}
    \UnaryInfC{$\typed[\objCtx]
      {\recP{\objCtx, y}{y} \inst \id_\objCtx}
      {\pTerm[(x : \bot_\objCtx)]{
          (\pAbstr{\objCtx}{C}) \inst \id_\objCtx}}$}
    \RightLabel{\Conv}
    \UnaryInfC{$\typed[\objCtx]
      {\recP{\objCtx, y}{y} \inst \id_\objCtx}{
        \pTerm[(x : \bot_\objCtx)]{C}}$}
    \DisplayProof
  \end{gather*}
  Note that the abstraction and instantiation steps consist of several
  steps, one for each variable in $\objCtx$.
  \qedDef
\end{example}

\begin{example}[Propositional Equality]
  In \iExRef{general-(co)product}, we remarked that the propositional equality
  arises as left adjoint to the contraction $\reidx{\delta}$, hence can
  be represented in the current type system.
  Let us elaborate this a bit more.
  First, note that the common way of implementing propositional equality in
  Agda is by means of the following data type.
  \begin{lstlisting}[language=Agda,mathescape=true,columns=flexible]
data $\Eq_A$ : $A \to A \to$ Set where
  $\refl$ : $(x : A)$ $\to$ $\Eq_A$ $x$ $x$
  \end{lstlisting}
  So it is clear that this type can be represented in our type system by
  \begin{equation*}
    \Eq_A(x,y)  \colonequals
    \indTy{X}{(x : A, y : A)}{\delta}{\top} \inst x \inst y,
  \end{equation*}
  where $\ctxMor{\delta}{(x : A)}{(x : A, y : A)}$ is the diagonal
  $\delta = (x, x)$.
  This type has the usual constructor
  $\refl \colonequals \inMu_1 : \pTerm[(x : A)]{\Eq_A(x,x)}$ and
  (weak) elimination rule
  \begin{gather*}
    \AxiomC{$\validTypeN[][x : A, y : A]{C}[]{i}$}
    \AxiomC{$\typed[x : A]{p}{C \substN{x/x, x/y}}$}
    \BinaryInfC{$\typed[x : A, y : A, z : \Eq_A(x,y)]{
        E^{\Eq}_{x,y}(p,z)}{C}$}
    \DisplayProof
  \end{gather*}
  This is of course not the full J-rule (or path induction) but it is already
  strong enough to prove, for example, the replacement (or substitution or
  transport) rule:
  \begin{gather*}
    \insertBetweenHyps{\hskip .1in}
    \AxiomC{$\validTypeN[][x : A]{P}[]{i}$}
    \AxiomC{$\typed[x : A]{p}{P}$}
    \AxiomC{$\typed[x, y : A]{t}{\Eq_A(x,y)}$}
    \TrinaryInfC{$\typed[y : A]{
        \repl(p,t)}{P \subst{y}{x}}$}
    \DisplayProof
  \end{gather*}
  by using $C = P\subst{y}{x}$, weakening this type to
  $\validTypeN[][x,y : A]{C}[]{i}$, and putting
  \begin{equation*}
    \repl(p,t) \colonequals E^{\Eq}_{x,y}(p,t).
    \tag*{\qedDef}
  \end{equation*}
\end{example}

\begin{example}[Primitive corecursion and successor]
  In \iExRef{extended-nats}, we claimed that the successor map
  $s_\infty : \pTerm[(y : \EN)]{\EN}$ on the extended naturals can be defined
  by means of primitive corecursion.
  Let us introduce this principle and the definition of $s_\infty$.

  First, by primitive corecursion (for $\EN$) we mean that for any $C$
  and $d$, we can find an $h$ as in the following diagram
  \begin{equation*}
    \begin{tikzcd}[column sep=2cm]
      C \rar{h} \dar{d} & \EN \dar{\outNu} \\
      (\oneT + C) + \EN \rar{\coprodArr{\id_\oneT + h, \outNu}}
      & \oneT + \EN
    \end{tikzcd}
  \end{equation*}
  We can derive this principle as follows.
  Note that if we define
  \begin{equation*}
    a = \coprodArr{
      \coprodArr{\kappa_1 \inst x, \kappa_2 \inst \kappa_1 \inst x},
      \kappa_2 \inst \kappa_2 \inst x
    }
  \end{equation*}
  we have $a : (\oneT + C) + \EN \ctxTo \oneT + (C + \EN)$.
  Now suppose we are given $\typed[\objCtxA, y : C]{d}{(\oneT + C) + \EN}$,
  if we define $d' = a \inst \, (\coprodArr{d,\kappa_2 \inst y} \inst y)$,
  then we find
  \begin{equation*}
    \AxiomC{$\typed[\objCtxA, y : \EN]{\kappa_2 \inst y}{(\oneT + C) + \EN}$}
    \UnaryInfC{$
      \typed[\objCtxA, y : C + \EN]{
        \coprodArr{d,\kappa_2 \inst y} \inst y}{(\oneT + C) + \EN}$}
    \UnaryInfC{$\typed[\objCtxA, y : C + \EN]{d'}{\oneT + (C + \EN)}$}
    \DisplayProof
  \end{equation*}
  This gives us, by $\EN$ being a coinductive type,
  \begin{equation*}
    \typed[\objCtxA]{\corecP{y}{d'}}{\pTerm[y:C + \EN]{\EN}}
  \end{equation*}
  and thus we can define $h = \corecP{y}{d'} \inst \, (\kappa_1 \inst y)$
  to get
  \begin{equation*}
    \typed[\objCtxA, y : C]{h}{\EN}.
  \end{equation*}

  Finally, we can use primitive corecursion to define $s_\infty$ by
  taking the extension of
  \begin{equation*}
    \coprodArr{\kappa_1 \circ \kappa_2 \circ \kappa_2, \kappa_2}
    : \EN + \EN \to (\oneT + (\EN + \EN)) + \EN,
  \end{equation*}
  giving us $h : \EN + \EN \to \EN$.
  Thus we put $s_\infty = h \circ \kappa_1$.
  \qedDef
\end{example}

\section{Meta Properties}
\label{app:basic-props}

\subsection{Basic Meta Properties}

\begin{proposition}
  \label{prop:derived-structural-rules}
  The following rules holds for the calculus given in \iSecRef{syntax}.
  \begin{itemize}
  \item Substitution
    \begin{gather*}
      \AxiomC{$\validTypeN[!][\objCtx_1, x:A, \objCtx_2]{B}[\objCtx_3]{i}$}
      \AxiomC{$\typed[\objCtx_1]{t}{A}$}
      \BinaryInfC{$
        \validTypeN[!][\objCtx_1, \objCtx_2 \subst{t}{x}]
        {B \subst{t}{x}}[\objCtx_3 \subst{t}{x}]{i}$}
      \DisplayProof
      \\[7pt]
      \AxiomC{$\typed[\objCtx_1, x:A, \objCtx_2]{s}{T}$}
      \AxiomC{$\typed[\objCtx_1]{t}{A}$}
      \BinaryInfC{
        $\typed[\objCtx_1, \objCtx_2 \subst{t}{x}]
        {s \subst{t}{x}}
        {T \subst{t}{x}}$}
      \DisplayProof
    \end{gather*}
  \item Exchange
     \begin{gather*}
      \AxiomC{$\validTypeN[!][\objCtx_1, x:A, y : B, \objCtx_2]{C}[\objCtx_3]{i}$}
      \AxiomC{$x \not \in \fv(B)$}
      \BinaryInfC{$
        \validTypeN[!][\objCtx_1, y : B, x:A, \objCtx_2]{C}[\objCtx_3]{i}$}
      \DisplayProof
      \\[7pt]
      \AxiomC{$
        \typed[\objCtx_1, x:A, y : B, \objCtx_2]{t}{\pTerm[\objCtx_3]{C}}$}
      \AxiomC{$x \not \in \fv(B)$}
      \BinaryInfC{$\typed[\objCtx_1, y:B, x:A, \objCtx_2]
        {t}{\pTerm[\objCtx_3]{C}}$}
      \DisplayProof
    \end{gather*}
  \item Contraction
    \begin{gather*}
      \AxiomC{$\validTypeN[!][\objCtx_1, x:A, y:A, \objCtx_2]{C}[\objCtx_3]{i}$}
      \UnaryInfC{$\validTypeN[!][\objCtx_1, x:A, \objCtx_2 \subst{x}{y}]
        {C \subst{x}{y}}[\objCtx_3 \subst{x}{y}]{i}$}
      \DisplayProof
      \\[7pt]
      \AxiomC{$\typed[\objCtx_1, x:A, y:A, \objCtx_2]{t}{\pTerm[\objCtx_3]{C}}$}
      \UnaryInfC{$\typed[\objCtx_1, x:A, \objCtx_2 \subst{x}{y}]
        {t\subst{x}{y}}{\pTerm[\objCtx_3 \subst{x}{y}]{C \subst{x}{y}}}$}
      \DisplayProof
    \end{gather*}
  \end{itemize}
\end{proposition}
\begin{proof}
  In each case, the rules are straightforwardly proved by simultaneous
  induction over types and terms.
  It should be noted that for types only the instantiation and weakening
  rules appear as cases, since the other rules have only types without
  free variables in the conclusion.
  Similarly, only terms constructed by means of the the projection, weakening
  or the instantiation rule appear as cases in the proofs.
\end{proof}

Analogously, the substitution, exchange and contraction rules for type variables
are valid in the calculus, as well.

\subsection{Subject Reduction}

\begin{proof}[Proof of \lemRef{correctness-type-action}]
  Recall that we have to prove
  \begin{equation*}
    \AxiomC{$
      \validTypeN[X: {\pType[\objCtx_1]}][\objCtx_2']{C}[\objCtx_2]{i}$}
    \AxiomC{$\typed[\objCtx_1, x : A]{t}{B}$}
    \BinaryInfC{$
      \typed[\objCtx_2', \objCtx_2, x : \tyFunc{C}(A)]
      {\tyFunc{C}(t)}{\tyFunc{C}(B)}$}
    \DisplayProof
  \end{equation*}
  We want to prove this by induction in the derivation of $C$,
  thus we need to generalise the statement to arbitrary type constructor
  contexts $\tyCtx$.
  So let
  $\tyCtx = X_1 : \pType[\objCtx_1], \dotsc, X_n : \pType[\objCtx_n]$
  be a context,
  $\validTypeN[!][\objCtx']{C}[\objCtx]{j}$
  a type and
  $\typed[\objCtx_i, x : A_i]{t_i}{B_i}$ terms for $i = 1, \dotsc, n$.
  We show that
  $\typed[\objCtx', \objCtx, x : \tyFunc{C}(\vec{A})]{\tyFunc{C}(\vec{t})}{\tyFunc{C}(\vec{B})}$
  holds by induction in the derivation of $C$.

  The induction base has two cases.
  First, it is clear that if $\tyCtx = \emptyset$, then $\vec{A} = \varepsilon$
  and $\tyFunc{C}(\vec{A}) = C$, thus the definition
  is thus well-typed.
  Second, if $C = X_i$ for some $i$, then we immediately have
  \begin{equation*}
    \tyFunc{C}(\vec{A})
    = C \tySubst{\vec{\pAbstr{\objCtx_i}{A}}}{\vec{X}} \, \id_{\objCtx_i}
    = (\pAbstr{\objCtx_i}{A_i}) \, \id_{\objCtx_i}
    \paramReduce A_i,
  \end{equation*}
  thus, by \rulelabel{Conv} and the type of $t_i$, we have
  $\typed[\objCtx_i, x : \tyFunc{C}(\vec{A})]{t_i}{\tyFunc{C}(\vec{B})}$ as
  required.

  In the induction step, we have five cases for $C$.
  \begin{itemize}
  \item The type correctness for $\tyFunc{C}$ in case $C$ has been constructed
    by Weakening for type and term variables is immediate by
    induction and the definition of $F$ in these cases.
  \item $C = C' \inst s$ and $\objCtx = \objCtxA \subst{s}{y}$ with
    \begin{equation*}
      \AxiomC{$\validTypeN[\tyCtx][\objCtx']
        {C'}[(y : D, \objCtxA)]{i}$}
      \AxiomC{$\typed[\objCtx']{s}{D}$}
      \BinaryInfC{$\validTypeN[\tyCtx][\objCtx']
        {C' \inst s}[\objCtxA \subst{s}{y}]{i}$}
      \DisplayProof
    \end{equation*}
    By induction we have then that
    $\typed[\objCtx', y : D, \objCtxA, x : \tyFunc{C'}(\vec{A})]
    {\tyFunc{C'}(\vec{t})}{\tyFunc{C'}(\vec{A})}$,
    thus, since
    \begin{align*}
      \tyFunc{C'}(\vec{A})
      & = C' \tySubst{\vec{\pAbstr{\objCtx_1}{A}}}{\vec{X}}
      \inst \id_{y : D, \objCtxA} \\
      & = C' \tySubst{\vec{\pAbstr{\objCtx_1}{A}}}{\vec{X}}
      \inst y \inst \id_{\objCtxA},
    \end{align*}
    we get by \iPropRef{derived-structural-rules}
    \begin{align*}
      & \objCtx', \objCtxA \subst{s}{y},
      x : C' \tySubst{\vec{\pAbstr{\objCtx_1}{A}}}{\vec{X}}
      \inst s \inst \id_{\objCtxA \subst{s}{y}} \\
      & \qquad \quad \vdash
      {\tyFunc{C'}(\vec{t})\subst{s}{y}}
      :
      {C' \tySubst{\vec{\pAbstr{\objCtx_1}{B}}}{\vec{X}}
        \inst s \inst \id_{\objCtxA \subst{s}{y}}}.
    \end{align*}
    As we now have
    \begin{align*}
      F_{C' \inst s}(\vec{A})
      & = (C' \inst s) \tySubst{\vec{\pAbstr{\objCtx_1}{A}}}{\vec{X}}
      \inst \id_{\objCtxA  \subst{s}{y}} \\
      & = C' \tySubst{\vec{\pAbstr{\objCtx_1}{A}}}{\vec{X}}
      \inst s \inst \id_{\objCtxA  \subst{s}{y}}
    \end{align*}
    and $F_{C' \inst s}(\vec{t}) = \tyFunc{C'}(\vec{t})\subst{s}{y}$,
    we find that
    \begin{equation*}
      \typed[\objCtx', \objCtxA \subst{s}{y},
      x : F_{C' \inst s}(\vec{A})]
      {F_{C' \inst s}(\vec{t})}
      {F_{C' \inst s}(\vec{B})}
    \end{equation*}
    as expected.
  \item $C = \pAbstr{y}{C'}$ with
    $\validTypeN[!][\Gamma', y : D]{C'}[\Gamma]{i}$.
    This gives us, by induction,
    $\typed[\Gamma', y : D, \Gamma, x : \tyFunc{C'}^\tyCtx(\vec{A})]
    {\tyFunc{C'}^\tyCtx(\vec{t})}{\tyFunc{C'}^\tyCtx(\vec{B})}$.
    Now we observe that
    \begin{align*}
      \tyFunc{C'}^\tyCtx(\vec{A})
      & = C' \tySubst{\vec{\pAbstr{\objCtx_1}{A}}}{\vec{X}} \inst \id_\objCtx \\
      & \paramReduceInv (
      \pAbstr{y}{C' \tySubst{\vec{\pAbstr{\objCtx_1}{A}}}{\vec{X}}})
      \inst y \inst \id_\objCtx \\
      & = C \tySubst{\vec{\pAbstr{\objCtx_1}{A}}}{\vec{X}}
      \inst \id_{y : D, \Gamma} \\
      & = F_{C}^\tyCtx(\vec{A}),
    \end{align*}
    which gives us, by \rulelabel{Conv}, that
    $\typed[\Gamma', y : D, \Gamma, x : F_{C}^\tyCtx(\vec{A})]
    {\tyFunc{C'}^\tyCtx(\vec{t})}{F_{C}^\tyCtx(\vec{B})}$.
    Thus the definition
    $F_{\pAbstr{x}{C'}}^\tyCtx(\vec{t}) = \tyFunc{C'}^\tyCtx(\vec{t})$ is
    well-typed.
    \vspace*{5pt}
  \item $C = \indTy{Y}{\objCtx}{\vec{\sigma}}{\vec{D}}$ with
    \begin{equation*}
      \AxiomC{$
        \validTypeN[\tyCtx, Y : {\pType[\objCtx]}][\objCtxA_k]
        {D_k}[]{i}$}
      \AxiomC{$\ctxMor{\sigma_k}{\objCtxA_k}{\objCtx}$}
      \BinaryInfC{$
        \validTypeN[\tyCtx][\emptyset]
        {\indTy{Y}{\objCtx}{\vec{\sigma}}{\vec{D}}}[\objCtx]{i}$}
      \DisplayProof
    \end{equation*}
    For brevity, we define
    $R_{\vec{B}} =
    \indTy{Y}{\objCtx}{\vec{\sigma}}{\vec{D} \tySubst{\vec{B}}{\vec{X}}}$.
    Then, by induction, we have
    \begin{equation*}
      \typed[\objCtx, x : F_{D_k}^{\tyCtx, Y}(\vec{A}, R_{\vec{B}})]
      {F_{D_k}^{\tyCtx,Y}(\vec{t}, x)}
      {F_{D_k}^{\tyCtx, Y}(\vec{B}, R_{\vec{B}})}
    \end{equation*}
    Now we note that
    $F_{D_k}^{\tyCtx, Y}(\vec{A}, R_{\vec{B}})
    = D_k \tySubst{\vec{\pAbstr{\objCtx_i}{A}}}{
      \vec{X}} \tySubst{R_{\vec{B}}}{Y}$.%
    \footnote{
      Note that the second substitution does not contain a parameter
      abstraction, as $R_{\vec{B}}$ is closed.}
    If we define
    \begin{equation*}
      g_k =
      \inMu_k
      \inst \id_{\objCtxA_k} \inst
      \left(F_{D_k}^{\tyCtx, Y}\left(
          \vec{t},\id_{R_{\vec{B}}}
        \right)\right),
    \end{equation*}
    where $\inMu_k$ refers to
    $\inMuT[Y]{\objCtx}{\vec{\sigma}}{\vec{D}
      \tySubst{\vec{\pAbstr{\objCtx_i}{B}}}{\vec{X}}}_k$
    (see the definition of $F$), then we can derive the following.
    \begin{equation*}
      \AxiomC{$
        \typed[]{\inMu_k}
        {\appCtx{\objCtxA_k}(F_{D_k}^{\tyCtx, Y}(\vec{A}, R_{\vec{B}})
          \ctxTo R_{\vec{B}} \inst \sigma_k)}$}
      \RightLabel{\Inst}
      \UnaryInfC{$
        \typed[\objCtxA_k]{\inMu_k \inst \id_{\objCtxA_k}}{
          F_{D_k}^{\tyCtx, Y}(\vec{A}, R_{\vec{B}})
          \ctxTo R_{\vec{B}} \inst \sigma_k}$}
      \RightLabel{\Inst}
      \UnaryInfC{$
        \typed[\objCtxA_k,
          x : D_k \tySubst{\vec{\pAbstr{\objCtx_i}{A}}}{
            \vec{X}} \tySubst{R_{\vec{B}}}{Y}]
        {g_k}{R_{\vec{B}}}$}
      \RightLabel{\rulelabel{Ind-E}}
      \UnaryInfC{$
        \typed{\rec \, x \, \vec{g} \inst \id_{\objCtx}}
        {R_{\vec{A}} \inst \id_{\objCtx} \ctxTo R_{\vec{B}} \inst \id_{\objCtx}}$}
      \RightLabel{\Inst}
      \UnaryInfC{$
        \typed[\objCtx, x : R_{\vec{A}} \inst \id_{\objCtx}]
        {\rec \, x \, \vec{g} \inst \id_{\objCtx} \inst x}
        {R_{\vec{B}} \inst \id_{\objCtx}}$}
      \DisplayProof
    \end{equation*}
    Finally, we have
    \begin{align*}
      \tyFunc{C}^{\tyCtx}(\vec{A})
      & = \indTy{Y}{\objCtx}{\vec{\sigma}}{\vec{D}}
        \tySubst{\vec{\pAbstr{\objCtx_i}{A}}}{\vec{X}} \inst \id_{\objCtx} \\
      & = R_{\vec{A}} \inst \id_{\objCtx},
    \end{align*}
    which implies, by the above derivitations, that we indeed have
    \begin{equation*}
      \typed[\objCtx, x : \tyFunc{C}^{\tyCtx}(\vec{A})]
      {\tyFunc{C}^{\tyCtx}(\vec{t})}
      {\tyFunc{C}^{\tyCtx}(\vec{B})}.
    \end{equation*}
  \item $C = \coindTy{Y}{\objCtx}{\vec{\sigma}}{\vec{D}}$.
    This case is treated analogously to that for inductive types.
  \end{itemize}
  This concludes the induction, thus \eqRef{typing-functor-from-type}
  indeed holds for all $C$, $\vec{A}$, $\vec{B}$ and $\vec{t}$.
\end{proof}

\section{Strong Normalisation}
\label{app:sn}

\subsection*{Pre-Types and -Terms}

\begin{definition}[Pre-Types]
  See \iFigRef{pre-types}
  \begin{figure*}
  \begin{gather*}
    \AxiomC{}
    \RightLabel{\rulelabel{PT}-$\oneT$}
    \UnaryInfC{$\validTypeN[][]{\oneT}[]{i}$}
    \DisplayProof
    \\[7pt]
    \AxiomC{}
    \RightLabel{\rulelabel{PT-TyVar}}
    \UnaryInfC{$
      \validType[\tyCtx, X : {\pType[\objCtx]} \mid \emptyset]{X}{
        \pType[\objCtx]}$}
    \DisplayProof
    \\[7pt]
    \AxiomC{$\validTypeN{A}{i}$}
    \RightLabel{\rulelabel{PT-TyWeak}}
    \UnaryInfC{$\validTypeN[\tyCtx, X : {\pType[\objCtx]}]{A}{i}$}
    \DisplayProof
    \\[7pt]
    \AxiomC{$\validTypeN{A}{i}$}
    \RightLabel{\rulelabel{PT-Weak}}
    \UnaryInfC{$\validTypeN[\tyCtx][\objCtx_1, x]{A}{i}$}
    \DisplayProof
    \\[7pt]
    \AxiomC{$\validType[\ctx{\tyCtx}{\objCtx_1}]
      {A}{\pType[(x, \objCtx_2)]}$}
    \AxiomC{$\preTerm[\objCtx_1]{t}$}
    \RightLabel{\rulelabel{PT-Inst}}
    \BinaryInfC{$\validTypeN{A \inst t}[\objCtx_2]{i}$}
    \DisplayProof
    \\[7pt]
    \AxiomC{$\validTypeN[\tyCtx][\objCtx_1, x : A]{B}{i}$}
    \RightLabel{\rulelabel{PT-Param-Abstr}}
    \UnaryInfC{$
      \validType[\ctx{\tyCtx}{\objCtx_1}]
      {\pAbstr{x}{B}}{\pType[(x, \objCtx_2)]}$}
    \DisplayProof
    \\[7pt]
    \AxiomC{$
      \validType[\ctx{\tyCtx, X : {\pType[\objCtx]}}{\objCtx_k}]
      {A_k}{\univ{i}}$}
    \AxiomC{$\ctxMor{\sigma_k}{\objCtx_k}{\objCtx}$}
    \AxiomC{$k = 1, \dotsc, n \qquad \rho \in \{\mu, \nu\}$}
    \RightLabel{\rulelabel{PT-FP}}
    \TrinaryInfC{$
      \validTypeN[!][\emptyset]{
        \recT{\rho}{X : \pType[\Gamma]}{\vec{\sigma}}{\vec{A}}}[\Gamma]{i}$}
    \DisplayProof
  \end{gather*}
    \caption{Pre-Types}
    \label{fig:pre-types}
  \end{figure*}
\end{definition}

\begin{definition}[Pre-Terms]
  See \iFigRef{pre-terms}.
  \begin{figure*}
  \begin{gather*}
    \AxiomC{}
    \RightLabel{\rulelabel{PO-$\oneT$-I}}
    \UnaryInfC{$\preTerm[]{\unit}$}
    \bottomAlignProof
    \DisplayProof
    \qquad
    \AxiomC{$\preTerm[\objCtx_1]{t}[(x, \objCtx_2)]$}
    \AxiomC{$\preTerm[\objCtx_1]{s}$}
    \RightLabel{\rulelabel{PO-Inst}}
    \BinaryInfC{$\preTerm[\objCtx_1]{t \inst s}[\objCtx_2]$}
    \bottomAlignProof
    \DisplayProof
    \\[7pt]
    \AxiomC{$x \in \var$}
    \RightLabel{\rulelabel{PO-Proj}}
    \UnaryInfC{$\preTerm[\objCtx, x]{x}$}
    \DisplayProof
    \qquad
    \AxiomC{$\preTerm[\objCtx_1]{t}[\objCtx_2]$}
    \RightLabel{\rulelabel{PO-Weak}}
    \UnaryInfC{$\preTerm[\objCtx_1, x]{t}[\objCtx_2]$}
    \DisplayProof
    \\[7pt]
    \AxiomC{$k \in \N$}
    \RightLabel{\rulelabel{PO-Ind-I}}
    \UnaryInfC{$\preTerm[]{\inMu_k}[\objCtx, x]$}
    \DisplayProof
    \quad
    \AxiomC{$k \in \N$}
    \RightLabel{\rulelabel{PO-Coind-E}}
    \UnaryInfC{$\preTerm[]{\outNu_k}[\objCtx, x]$}
    \DisplayProof
    \\[7pt]
    \AxiomC{$\preType[][]{\indTy{X}{\objCtx}{\vec{\sigma}}{\vec{A}}}[\objCtx]$}
    \AxiomC{$\forall k = 1, \dotsc, n. \, \left( \preTerm[\objCtxA, \objCtx_k, y_k]{g_k} \right)$}
    \RightLabel{\rulelabel{PO-Ind-E}}
    \BinaryInfC{$
      \preTerm[\objCtxA]{
        \rec^{\indTy{X}{\objCtx}{\vec{\sigma}}{\vec{A}}} \vec{\pAbstr{\objCtx_k, y_k}{N_k}}}
      [(\objCtx, y)]$}
    \DisplayProof
    \\[7pt]
    \AxiomC{$\preType[][]{\coindTy{X}{\objCtx}{\vec{\sigma}}{\vec{A}}}[\objCtx]$}
    \AxiomC{$\forall k = 1, \dotsc, n. \, \left( \preTerm[\objCtxA, \objCtx_k, y_k]{g_k} \right)$}
    \RightLabel{\rulelabel{PO-Coind-I}}
    \BinaryInfC{$
      \preTerm[\objCtxA]{
        \corec^{\coindTy{X}{\objCtx}{\vec{\sigma}}{\vec{A}}} \vec{\pAbstr{\objCtx_k, y_k}{N_k}}}
      [(\objCtx, y)]$}
    \DisplayProof
  \end{gather*}
    \caption{Pre-Terms}
    \label{fig:pre-terms}
  \end{figure*}
\end{definition}

\begin{remark}
  The intuition for \iDefRef{type-action} can be better understood in terms of
  the diagrams that correspond to, for example, the definition on initial
  dialgebras.
  Put
  $R_{\vec{A}} = \indTy{Y}{\objCtx}{\vec{\sigma}}{\vec{D}
    \tySubst{\vec{\pAbstr{\objCtx_i}{A}}}{\vec{X}}}$
  and analogous for $R_{\vec{B}}$.
  Then $\tyFunc{\indTy{Y}{\objCtx}{\vec{\sigma}}{\vec{D}}}(\vec{t})$ is defined
  as the morphism $h$ in the following diagram.
  \begin{equation*}
    \begin{tikzcd}[column sep=2cm]
       \tyFunc{D_k} \left(\vec{A}, R_{\vec{A}} \right)
       \rar{\tyFunc{D_k}\left(\vec{\id}, h\right)}
       \arrow{dd}{\inMu_k}
       & \tyFunc{D_k}\left(\vec{A}, R_{\vec{B}} \right)
       \dar{\tyFunc{D_k}(\vec{t}, \id)} \\
       & \tyFunc{D_k}\left(\vec{B}, R_{\vec{B}} \right)
       \dar{\alpha_k} \\
       R_{\vec{A}} \inst \sigma_k
       \rar{h [\sigma_k]}
       & R_{\vec{B}} \inst \sigma_k
    \end{tikzcd}
  \end{equation*}
\end{remark}

\subsection{Soundness proof for saturated sets model}

\begin{lemma}
  \label{lem:interpret-ctx-mor-composition}
  For all $\ctxMor{\sigma}{\objCtx_1}{\objCtx_2}$ and
  $\ctxMor{\tau}{\objCtx_2}{\objCtx_3}$ we have
  $\sem{\tau \bullet \sigma} = \sem{\tau} \circ \sem{\sigma}$.
\end{lemma}
\begin{proof}
  For all $\rho : \vars{\Gamma_1} \to \Terms$ and $x \in \vars{\objCtx_3}$
  we have
  \begin{align*}
    \sem{\tau \bullet \sigma}(\rho)(x)
    & = (\tau \bullet \sigma)(x)[\rho]
    = \tau(x)[\sigma][\rho] \\
    & = \sem{\tau}(\sem{\sigma}(\rho))(x)
    = (\sem{\tau} \circ \sem{\sigma})(\rho)(x)
  \end{align*}
  as required.
\end{proof}

\begin{lemma}
  \label{lem:move-inst-to-valuation}
  If $\validTypeN[]{A}{i}$,
  $\ctxMor{\sigma}{\objCtx_1}{\objCtx_2}$
  and $\rho \in \sem{\objCtx_1}$, then
  $\sem{A \inst \sigma}(\rho)
  = \sem{A}(\coprodArr{\rho, \sem{\sigma}(\rho)})$,
  where $\coprodArr{\rho, \sem{\sigma}(\rho)} \in \sem{\objCtx_1, \objCtx_2}$ is
  given by
  \begin{equation*}
    \coprodArr{\rho, \sem{\sigma}(\rho)}(x) =
    \begin{cases}
      \rho(x), & x \in \vars{\objCtx_1} \\
      \sem{\sigma}(\rho)(x), & x \in \vars{\objCtx_2}
    \end{cases}.
  \end{equation*}
\end{lemma}
\begin{proof}
  Simply by repeatedly applying the case of the semantics of type
  instantiations.
\end{proof}

The following four lemmas
\ref{lem:interpret-subst-lemma}-\ref{lem:monotonicity-interpretation}
are easily proved by induction in the derivation of the corresponding type $A$.
\begin{lemma}
  \label{lem:interpret-subst-lemma}
  If $\validTypeN[][\objCtx_1, x : B, \objCtx_2]{A}[]{i}$ and
  $\typed[\objCtx_1]{t}{B}$, then
  for all $\rho \in \sem{\objCtx_1, \objCtx_2 \subst{t}{x}}$ we have
  $\sem{A\subst{t}{x}} = \sem{A}(\rho[x \mapsto t])$.
\end{lemma}

\begin{lemma}
  \label{lem:interpret-move-subst}
  If $\validTypeN[\tyCtx][\objCtx_2]{A}[]{i}$ and
  $\ctxMor{\sigma}{\objCtx_1}{\objCtx_2}$, then
  for all $\delta \in \sem{\tyCtx}$ and $\rho \in \sem{\objCtx_1}$ we have
  $\sem{A [\sigma]}(\delta, \rho) = \sem{A}(\delta, \sem{\sigma}(\rho))$.
\end{lemma}

\begin{lemma}
  \label{lem:interpret-move-ty-subst}
  If $\validTypeN[\tyCtx_1, X : {\pType[\objCtx]}, \tyCtx_2]{A}[]{i}$
  and $\validTypeN[][]{B}[\objCtx]{i}$, then
  $\sem{A \tySubst{B}{X}}(\delta, \rho)
  = \sem{A}(\delta [X \mapsto \sem{B}], \rho)$.
\end{lemma}

\begin{lemma}
  \label{lem:monotonicity-interpretation}
  If $\validTypeN[\tyCtx][\objCtx]{A}[]{i}$ and
  $\delta, \delta' \in \sem{\tyCtx}$ with $\delta \sqsubseteq \delta'$
  (point-wise order), then for all $\rho \in \sem{\tyCtx}$
  \begin{equation*}
    \sem{A}(\delta, \rho) \subseteq \sem{A}(\delta', \rho).
  \end{equation*}
\end{lemma}

\begin{lemma}
  \label{lem:constructor-closure}
  Let $\mu = \indTy{X}{\objCtx}{\vec{\sigma}}{\vec{A}}$ where we have
  $\validTypeN[!][\emptyset]{\mu}[\objCtx]{i}$.
  If $\delta \in \sem{\tyCtx}$, $\rho \in \sem{\objCtx_k}$ and
  $P \in \sem{A_k}(\delta[X \mapsto \sem{\mu}(\delta)], \rho)$, then
  \begin{equation*}
    \inMu_k \inst \rho \inst P \in \sem{\mu}(\delta, \sem{\sigma_k}(\rho)).
  \end{equation*}
\end{lemma}
\begin{proof}
  Let $\delta$, $\rho$ and $P$ be given as in the lemma, and
  put $M = \inMu_k \inst \rho \inst P$.
  We need to show for any choice of $U \in I_\objCtx$ and
  \begin{equation*}
    N_k \in \compr{\sem{\objCtx_k}, \sem{A_k^\objCtxA}(\delta[X \mapsto U])}_y
    \Vdash \reidx{\sem{\sigma_k \bullet \pi}}(U)
  \end{equation*}
  that
  \begin{equation*}
    K = \rec \vec{\pAbstr{\objCtx_k, y}{N_k}}
    \inst \, (\sigma_k \bullet \rho) \inst M
  \end{equation*}
  is in $U(\sem{\sigma_k}(\rho))$.
  Now we define $r = \rec \vec{\pAbstr{\objCtx_k, y}{N_k}}$ and
  \begin{equation*}
    K' =
    N_k \subst*{\tyFunc{A_k}(r \inst \id_\objCtx \inst x')}{y} \substN*{\rho, P}
  \end{equation*}
  so that $K \contract K'$.
  Let us furthermore put
  \begin{equation*}
    V = U(\sem{\sigma_k}(\rho)).
  \end{equation*}
  By $V \in \sat$, it suffices to prove
  that $K' \in V$.
  Note that we can rearrange the substitution in $K'$ to get
  $K' = N_k \substN{\rho, P'}$ with
  $P' = \tyFunc{A_k}(r \inst \id_\objCtx \inst x') \substN{\rho, P}$.

  We get $K' \in V$ from
  $N_k \in \compr{\sem{\objCtx_k}, \sem{A_k}(\delta[X \mapsto U])}_y
  \Vdash \reidx{\sem{\sigma_k \bullet \pi}}(U)$,
  provided that $\rho \in \sem{\objCtx_k}$ and
  $P' \in \sem{A_k}(\delta[X \mapsto U], \rho)$.
  The former is given from the assumption of the lemma.
  The latter we get from \iLemRef{ty-functor-sound},
  since we have assumed soundness for the components of $\mu$ and
  $P \in \sem{\tyFunc{A_k}}(\rho)$.
  Thus we have $K' = N_k \substN{\rho, P'} \in V$.

  So by saturation we have $K \in V = U(\sem{\sigma_k}(\rho))$ for any choice
  of $U$ and $N_k$, thus if follows that
  $M \in \sem{\mu}(\delta, \sem{\sigma_k}(\rho))$.
\end{proof}

\begin{lemma}
  \label{lem:coind-is-largest-fp}
  Let $\nu = \coindTy{X}{\objCtx}{\vec{\sigma}}{\vec{A}}$ where we have
  $\validTypeN[!][\emptyset]{\nu}[\objCtx]{i}$.
  If $U \in I_{\objCtx}$ and $\delta \in \sem{\tyCtx}$, such that for all
  $M \in U(\rho)$, all $k$,
  all $\rho \in \sem{\objCtx}$ and all
  $\gamma \in \sem{\sigma_k}^{-1}(\rho)$,
  $\outNu_k \inst \gamma \inst M \in \sem{A_k}(\delta[X \mapsto U], \gamma)$,
  then
  \begin{equation*}
    \forall \rho. \, U(\rho) \subseteq \sem{\nu}(\delta, \rho).
  \end{equation*}
\end{lemma}
\begin{proof}
  This follows immediately from the definition of $\sem{\nu}$, just instantiate
  the definition with the given $U$.
  Then all all $M \in U(\rho)$ are in $\sem{\nu}(\delta, \rho)$.
\end{proof}

\begin{lemma}
  \label{lem:corec-closure}
  Let $\nu = \coindTy{X}{\objCtx}{\vec{\sigma}}{\vec{A}}$ where we have
  $\validTypeN[!][\emptyset]{\nu}[\objCtx]{i}$.
  If $\delta \in \sem{\tyCtx}$, $U \in I_{\objCtx}$,
  $\rho \in \sem{\objCtx_k}$ and
  \begin{equation*}
     N_k \in
     \compr{\sem{\objCtx_k}, \reidx{\sem{\sigma_k}}(U)}_y
     \Vdash \sem{A_k}(\delta[X \mapsto U])
  \end{equation*}
  for $k = 1, \dotsc, n$,
  then
  \begin{equation*}
    \corecPs{\objCtx_k, y}{N_k} \inst \rho \inst M
    \in \sem{\nu}(\delta, \rho).
  \end{equation*}
\end{lemma}
\begin{proof}
  Similar to the proof of \iLemRef{constructor-closure} by using that
  the interpretation of $\nu$-types is a largest fixed
  point~\iLemRef{coind-is-largest-fp} and that the interpretation is
  monotone~\iLemRef{monotonicity-interpretation}.
\end{proof}

\begin{lemma}
  \label{lem:interpret-resp-reduction}
  Suppose $C$ is a type with $\validTypeN{C}{i}$.
  If  $\rho, \rho' \in \sem{\objCtx_1,\objCtx_2}$ with
  $\rho \reduce \rho'$, then
  $\forall \delta. \,
  \sem{C}(\delta, \rho) = \sem{C}(\delta, \rho')$.
  Furthermore, if $C \reduce C'$, then $\sem{C} = \sem{C'}$.
\end{lemma}
\begin{proof}
  The first part follows by an easy induction, in which the only interesting
  case $C = X$ is.
  Here we have
  \begin{align*}
    \sem{C}(\delta, \rho)
    & = \delta(X)(\rho) \\
    & = \delta(X)(\rho') \\
    & = \sem{C}(\delta, \rho'),
  \end{align*}
  since $\delta(X) \in I_{\objCtx_i}$ and thus respects conversions.

  For the second part, let $D$ be given by replacing all terms in parameter
  position in $C$ by variables, so that $C = D[\rho]$ for some substitution
  $\rho$.
  But then there is a $\rho'$ with $\rho \reduce \rho'$ and $C' = D[\rho']$,
  and the claim follows from the first part.
\end{proof}

\begin{proof}[Proof of \lemRef{ty-functor-sound}]
  We proceed by induction in the derivation of
  $\validTypeN[\tyCtx][\objCtx]{C}[\objCtx']{i}$.
  \begin{itemize}
  \item $\validTypeN[][]{\oneT}[]{i}$ by \TTyI.
    In this case we have that $\tyFunc{\oneT}(\varepsilon) = x \in \base$,
    thus $\tyFunc{\oneT}(\varepsilon) \in \sem{C}(\delta)$ by
    by saturation.
  \item $\validType[\tyCtx, X_{n+1} : {\pType[\objCtx_{n+1}]} \mid \emptyset]{X}{
      \pType[\objCtx_{n+1}]}$ by \TyVarI.
    Note that
    $\tyFunc{X_{n+1}}(\vec{t}, t_{n+1}) = t_{n+1}$ and
    $\sem{X_{n+1}}(\delta, \sigma) = \delta(X_{n+1})(\sigma)
    = \sem{B_{n+1}}(\sigma)$.
    thus the claim follows directly from the assumption of the lemma.
  \item $\validTypeN[\tyCtx, X : {\pType[\objCtx'']}][\objCtx]{C}[\objCtx']{i}$
    by \TyVarWeak.
    Immediate by induction.
  \item $\validTypeN[\tyCtx][\objCtx, y : D]{C}[\objCtx']{i}$ by \TyWeak.
    Again immediate by induction.
  \item $\validTypeN[!][\objCtx]{C \inst r}[\objCtx' \subst{s}{x}]{i}$ with
    $\typed[\objCtx]{r}{C'}$ by \TyInst.
    First, we note that $\sigma = (\sigma_1, \sigma_2)$ with
    $\ctxMor{\sigma_1}{\objCtxA}{\objCtx}$ and
    $\ctxMor{\sigma_2}{\objCtxA}{\objCtx'[\sigma_1,r]}$.
    Let us put $\tau = (\sigma_1, r \substN{\sigma_1}, \sigma_2)$,
    so that we have
    \begin{align*}
      \tyFunc{(C \inst r)}(\vec{t}) \substN{\sigma, s}
      & = \tyFunc{C}(\vec{t}) \subst{s}{x} \substN{\sigma, s} \\
      & = \tyFunc{C}(\vec{t})
      \substN{\sigma_1, r \substN{\sigma_1}, \sigma_2, s} \\
      & = \tyFunc{C}(\vec{t}) \substN{\tau, s}.
    \end{align*}
    By the assumption of the lemma on parameters we have
    $r \substN{\sigma_1} \in \sem{C'}(\sigma_1)$, and thus
    $\tau \in
    \sem*{\objCtx, x : C' \substN{\sigma_1},
      \objCtx' \substN{\sigma_1, r \substN{\sigma_1}}}$,
    which gives
    $\sem{C \inst r}(\delta, \sigma) = \sem{C}(\delta, \tau)$.
    By induction, we have
    $\tyFunc{C}(\vec{t}) \substN{\tau, s} \in \sem{C}(\delta, \tau)$,
    and
    \begin{equation*}
      \tyFunc{(C \inst r)}(\vec{t}) \substN{\sigma, s}
      \in \sem{C \inst r}(\delta, \sigma)
    \end{equation*}
    follows.
  \item $\validType[\ctx{\tyCtx}{\objCtx_1}]{
      \pAbstr{x}{B}}{\pType[(x : A, \objCtx_2)]}$
    by \PAbstr.
    Immediate by induction.
  \item $\validTypeN[!][\emptyset]{
      \indTy{Y}{\Gamma}{\vec{\tau}}{\vec{D}}}[\Gamma]{i}$
    by \FPTy.
    We abbreviate, as before, this type just by $\mu$.
    Recall the definition of $\tyFunc{\mu}$:
    \begin{align*}
      & \tyFunc{\mu}(\vec{t})
      = \recPs[R_A]{\objCtxA_k, x}{g_k} \inst \id_{\objCtx} \inst x \\
      & \quad \text{with } g_k =
      \inMu_k
      \inst \id_{\objCtxA_k} \inst
      \left(\tyFunc[\tyCtx, Y]{D_k}\left(
          \vec{t}, y
        \right)\right) \\
      & \quad \text{and }
      R_A = \mu \tySubst{\vec{\pAbstr{\objCtx_i}{A_i}}}{\vec{X}} \\
      & \quad \text{for }
      \validTypeN[\tyCtx, Y : {\pType[\objCtx]}][\Delta_k]{D_k}[]{i}
    \end{align*}
    Now, put $\delta' = \delta [Y \mapsto \sem{R_B}]$,
    then we have by induction that
    $\tyFunc{D_k}(\vec{t}, y) \substN{\id_{\objCtxA_k}, y}
    \in \sem{D_k}(\delta', \id_{\objCtxA_k})$.
    Since $\id_{\objCtx_k} \in \SN$ we have by \iLemRef{constructor-closure}
    for all $\rho \in \sem{\objCtxA_k, \tyFunc{D_k}(\vec{A}, R_B)}$ that
    $g_k [\rho] \in \sem{\mu}(\delta, \sem{\tau_k}(\rho))$.
    By assumption, we have $s \in \sem{R_A}(\sigma)$, hence by
    choosing $U = \sem{R_B}$ in the definition of $\sem{R_A}$ we find
    $\tyFunc{\mu}(\vec{t})[\sigma,s] \in
    \sem{R_B}(\sigma) = \sem{\mu}(\delta, \sigma)$.
  \item $\validTypeN[!][\emptyset]{
      \coindTy{Y}{\Gamma}{\vec{\tau}}{\vec{D}}}[\Gamma]{i}$
    by \FPTy.
    Analagous to the inductive case, only that we use \iLemRef{corec-closure}.
  \end{itemize}
\end{proof}

\begin{proof}[Proof of \lemRef{sat-sound}]
  We proceed by induction in the type derivation for $t$.
  Since $t$ does not have any parameters, we only have to deal with fully
  applied terms and will thus leave out the case for instantiation in
  the induction.
  Instead, we will have cases for fully instantiated $\inMu$, $\outNu$, etc.
  So let $\rho \in \sem{\objCtx}$ and proceed by the cases for $t$.
  \begin{itemize}
  \item $\unit \in \sem{\oneT}(\rho)$ by definition.
  \item For $\typed[\objCtx, x : A]{x}{A}$ we have $x[\rho] = \rho(x)$.
    By definition of $\sem{\objCtx, x : A}$, we have
    $\rho(x) \in
    \sem{\validTypeN[][\objCtx]{A}[]{i}}(\restr{\rho}{\vars{\objCtx}})
    = \sem{\validTypeN[][\objCtx, x : A]{A}[]{i}\, }(\rho)$.
    Thus $x[\rho] \in \sem{A}(\rho)$ as required.
  \item Weakening is dealt with immediately by induction.
  \item If $t$ is of type $B$ by \rulelabel{Conv}, then by
    induction $t \in \sem{A}(\rho)$.
    Since by \iLemRef{interpret-resp-reduction} $\sem{B}(\rho) = \sem{A}(\rho)$,
    we have $t \in \sem{B}(\rho)$.
  \item Suppose we are given
    $\mu = \indTy{X}{\objCtx}{\vec{\sigma}}{\vec{A}}$ and
    $\typed[\objCtxA]{\inMu_k \inst \tau \inst t}{\mu [\sigma_k \bullet \tau]}$
    with $\ctxMor{\tau}{\objCtxA}{\objCtx}$ and
    $\typed[\objCtxA]{t}{A_k \tySubst{\mu}{X}[\tau]}$.

    Then, by induction, we have
    $t \in \sem{A_k}(X \mapsto \sem{\mu}, \tau)$
    and soundness for the components of $\mu$,
    thus by \iLemRef{constructor-closure}
    \begin{equation*}
      \inMu_k \inst \tau \inst t \in \sem{\mu}(\sem{\sigma_k \bullet \tau}(\rho))
    = \sem{\mu \substN{\sigma_k \bullet \tau}}(\rho).
    \end{equation*}


  \item Suppose we have
    $\mu = \indTy{X}{\objCtx}{\vec{\sigma}}{\vec{A}}$ and
    $\typed[\objCtxA]{\recPs[\mu]{\objCtx_k, x}{g_k} \inst \tau \inst t}{
      C \inst \tau}$.
    Then by induction we get from $\typed[\objCtxA]{t}{\mu \substN{\tau}}$
    that $t \in \sem{\mu \substN{\tau}}(\rho)$, hence if we chose
    $U = \sem{C \inst \tau}$ and $N_k = g_k$
    the definition of $\sem{\mu}$ yields
    $\recPs[\mu]{\objCtx_k, x}{g_k} \inst \tau \inst t
    \in \sem{C \inst \tau}(\rho)$.
  \item Suppose
    $\nu = \coindTy{X}{\objCtx}{\vec{\sigma}}{\vec{A}}$ and
    $\typed[\objCtxA]{\outNu_k \inst \tau \inst t}{
      A_k \tySubst{\nu}{X} \substN{\tau}}$
    with $\ctxMor{\tau}{\objCtxA}{\objCtx_k}$
    and $\typed[\objCtxA]{t}{\nu \substN{\sigma_k \bullet \tau}}$.
    By induction, $t \in \sem{\nu \substN{\sigma_k \bullet \tau}}(\rho)$
    thus there is a $U$ such that
    $\outNu_k \inst \tau \inst t \in \sem{A_k}(X \mapsto U, \sem{\tau}(\rho))$.
    By \iLemRef{coind-is-largest-fp} and \iLemRef{monotonicity-interpretation}
    we then have
    $\outNu_k \inst \tau \inst t
    \in \sem{A_k}(X \mapsto \sem{\nu}, \sem{\tau}(\rho))$.
    Since
    $\sem{A_k}(X \mapsto \sem{\nu}, \sem{\tau}(\rho))
    = \sem{A_k \tySubst{\nu}{X} \substN{\tau}}(\rho)$,
    the claim follows.
  \item For $\corec$-terms we just apply \iLemRef{corec-closure}, similar to
    the $\inMu_k$-case.
  \end{itemize}
  This concludes the induction, thus the interpretation of types is sound
  with respect to the typing judgement for terms.
\end{proof}


\end{document}